\documentclass[11pt,letter]{article}
\usepackage[papersize={8.5in,11in},margin=1in]{geometry}
\usepackage{amsthm} 
\usepackage{fancyhdr} 
\usepackage{color} 
\usepackage{changepage}
\usepackage[nofillcomment,linesnumbered,noend,noresetcount,noline]{algorithm2e}

\usepackage{amsmath} 
\usepackage{amssymb} 
\usepackage{graphicx} 
\usepackage{complexity} 
\usepackage{enumitem} 
\usepackage{shorttoc} 
\usepackage{array} 
\usepackage{float}
\usepackage{pdfsync}
\usepackage{verbatim}

\usepackage[pagebackref=true,colorlinks]{hyperref}
\hypersetup{linkcolor=red,filecolor=red,citecolor=red,urlcolor=red}

\newtheoremstyle{slplain}
  {.4\baselineskip\@plus.1\baselineskip\@minus.1\baselineskip}
  {.3\baselineskip\@plus.1\baselineskip\@minus.1\baselineskip}
  {\itshape}
  {}
  {\bfseries}
  {.}
  { }
  {}
\theoremstyle{slplain} 

\newtheorem*{definition*}{Definition}
\newtheorem*{theorem*}{Theorem}
\newtheorem{theorem}{Theorem}[section]
\newtheorem{lemma}[theorem]{Lemma}
\newtheorem{proposition}[theorem]{Proposition}

\newtheorem{corollary}[theorem]{Corollary}
\newtheorem{definition}[theorem]{Definition}

\makeatletter
\newtheorem*{rep@theorem}{\rep@title}
\newcommand{\newreptheorem}[2]{%
\newenvironment{rep#1}[1]{%
 \def\rep@title{#2 \ref{##1}}%
 \begin{rep@theorem}}%
 {\end{rep@theorem}}}
\makeatother

\newreptheorem{theorem}{Theorem}
\newreptheorem{lemma}{Lemma}
\newreptheorem{claim}{Claim}
\newreptheorem{corollary}{Corollary}
\newreptheorem{proposition}{Proposition}

\newenvironment{proofof}[1]{\par{\noindent \bf Proof of #1.}}{\qed\par}

\theoremstyle{definition}

\theoremstyle{remark}

\numberwithin{equation}{section}

\newtheoremstyle{etplain}
  {.0\baselineskip\@plus.1\baselineskip\@minus.1\baselineskip}
  {.0\baselineskip\@plus.1\baselineskip\@minus.1\baselineskip}
  {\itshape}
  {}
  {\bfseries}
  {.}
  { }
  {}

\newcommand{\namedref}[2]{\hyperref[#2]{#1~\ref*{#2}}}

\newcommand{\figurerefb}[2]{\hyperref[#1]{Figure~\ref*{#1}#2}}

\newcommand{\equationref}[1]{\hyperref[#1]{(\ref*{#1})}}
\renewcommand{\eqref}{\equationref}

\usepackage[textsize=tiny]{todonotes}

\newcommand{\DEBUG}[1]{}

\newcommand{\argmax}{\operatornamewithlimits{arg\,max}}
\newcommand{\argmin}{\operatornamewithlimits{arg\,min}}
\newcommand{\B}{\mathcal{B}}
\renewcommand{\M}{\mathcal{M}}
\newcommand{\rank}{\operatornamewithlimits{rank}}
\newcommand{\eql}{\texttt{eq}}
\newcommand{\floor}[1]{\lfloor {#1} \rfloor}
\newcommand{\ceil}[1]{\lceil {#1} \rceil}

\renewcommand{\setminus}{-}
\renewcommand{\emptyset}{\varnothing}

\newcommand{\vol}{\mathrm{vol}}

\newcommand{\Z}{\ComplexityFont{Z}}
\renewcommand{\C}{\mathcal{C}}

\newcommand{\vc}[1]{\mathbf{#1}}
\newcommand{\dbracket}[1]{[\![{#1}]\!]}

\def\abs#1{\left|#1  \right|}
\def\norm#1{\left\| #1 \right\|}

\newcommand\sw{\textsc{Sw}}
\renewcommand\Z{\mathbb{Z}}
\renewcommand\R{\mathbb{R}}
\renewcommand\E{\mathbb{E}}
\renewcommand\comment[1]{}

\begin{document}
 

\title{Computing Walrasian Equilibria:\\ Fast Algorithms and Structural
Properties}
\author{Renato Paes Leme \\ Google Research NY \\ \texttt{renatoppl@google.com}
\and Sam Chiu-wai Wong \\ UC Berkeley \\ \texttt{samcwong@berkeley.edu}}
\date{\today}
\date{}
\maketitle
\begin{abstract}
We present the first polynomial time algorithm for computing Walrasian equilibrium in an economy with
indivisible goods and \emph{general} buyer valuations having only access to an
\emph{aggregate demand oracle}, i.e., an oracle that given prices on all goods,
returns the aggregated demand over the entire population of buyers. For the important special case of gross
substitute valuations, our algorithm queries the aggregate demand oracle
$\widetilde{O}(n)$ times and takes $\widetilde{O}(n^3)$ time, where $n$ is the
number of goods. At the heart of our solution is a method for exactly minimizing certain convex functions
which cannot be evaluated but for which the subgradients can be computed.

We also give the fastest known algorithm for computing Walrasian equilibrium for
gross substitute valuations in the \emph{value oracle model}. Our algorithm
has running time $\widetilde{O}((mn + n^3) T_V)$ where $T_V$ is the cost of
querying the value oracle. A key technical ingredient is to regularize a convex
programming formulation of the problem in a way that subgradients are cheap to
compute. En route, we give necessary and sufficient
conditions for the existence of \emph{robust Walrasian prices}, i.e., prices for which each agent has a unique demanded bundle and the demanded
bundles clear the market. When such prices exist, the market can be perfectly
coordinated by solely using prices.
\end{abstract}

\thispagestyle{empty}

\setcounter{page}{1}
\section{Introduction}
 
\subsection{A macroscopic view of the market}

As part of our everyday experience, prices reach equilibria in a wide range of
economics settings. Yet, markets are complicated and consist of heterogeneous
goods and a huge population of buyers can have very diverse preferences
that are hard to model analytically. With the sheer amount of information needed
to describe the economy, how can the market possibly reach an equilibrium? 
\emph{Well, perhaps not all this information is needed.}

In this paper, we provide evidence supporting this belief through the lens of
algorithms. Specifically, we propose algorithms for computing market equilibrium
using very limited amount of information. Our result suggests that information
theoretically, it is not necessary to make too many measurements or observations
of the market to compute an equilibrium. This may also shed light into how
markets operate.

As the first step, we must design a realistic model to represent the economy.
The standard TCS approach would require the entire input be specified but for a
market, it is simply too computationally expensive to model its individual
agents in full details. So what should we turn to? If equilibrium represents the
collective behavior of the agents, perhaps some kind of aggregate information
would be enough. Such information can be average salaries, interest rate,
population, fashion trend and so on. An algorithm would ideally process these
\textit{macroscopic-scale} information in an efficient manner to compute
equilibrium price that allows the market to clear.

We show that it is possible to compute market equilibrium by exploiting the very rudimentary information of \emph{aggregate demand},
i.e. the quantity demanded for each item at a given price aggregated over the entire
population of buyers. This result
implies, among other things, that a market can be viewed as an aggregate
entity. For the sake of reaching equilibrium, detailed knowledge about its
individual buyers at the microscopic level may not really be needed.
Rather, it should be their collective behavior that dictates the outcome
of the market.

The use of aggregate demand by our algorithm also resonates with a common perception of
the role played by excess demand/supply. A highly sought-after good
would usually see its price soar whereas an unpopular good would be
inexpensive. This is similar to our algorithms which, in some sense, operate by increasing
the price of overdemanded good and vice versa in an iterative fashion. We note however that by no means are we suggesting that our algorithms closely
mirror how a market actually works. While the holy grail of this
research direction is to understand how a market reaches equilibrium
in practice, perhaps a humble first step is to show that this can
be done algorithmically with as little information and assumption
as possible.

Our starting point is the Gul and Stachetti's model \cite{GulStachetti} of an economy of
indivisible goods, but we make no further assumptions on the structure of the
valuation functions. The goal is to compute market equilibrium: a set of
item prices and allocations of items to buyers such that the market clears and
each buyer gets his or her favorite bundle of goods under the current prices.

The market can only be accessed via an \emph{aggregate demand oracle}: given
prices for each item, what is the demand for each item aggregated over the
entire population. Clearly in this model, it is not possible to compute an
allocation of items to buyers, since the oracle access model doesn't allow any
sort of buyer-specific information. Curiously, equilibrium prices are still
computable in a very efficient manner:

\begin{theorem*}[informal]
In a consumer market with $n$ goods and $m$ buyers, we can find
(a vector of) equilibrium prices, whenever it exists, using $\widetilde{O}(n^2)$
calls to the aggregate demand oracle and $\widetilde{O}(n^{5})$ time. If
valuations are gross substitutes,  $\widetilde{O}(n)$ calls to the aggregate
demand oracle and $\widetilde{O}(n^{3})$ time suffice.
\end{theorem*}

Notably, the number of buyers plays no role. Our algorithm has query and
time complexity essentially independent of the number of players.
This feature is especially relevant in practice
as markets are usually characterized by a large population and relatively
few number of goods. The city of Berkeley, for example, has about
350 restaurants but 120,000+ people!

\subsection{From telescopes to augmenting lenses}

Aggregate demand oracles are like looking at the economy from a telescope.
Having a telescope has its advantages: it is possible to get a very
global view of the economy with a few queries. On the other hand, extracting
details is hard.

Our second question is how fast equilibrium can be computed with only a
\emph{local view of the economy}? Our analogue for augmenting lenses will be
the \emph{value oracle model}, in which one can query the value of each buyer for each
bundle of items. This again has its advantages: it provides very fine-grained
information about the market, but has the shortcoming that many
queries are needed to extract any sort of global information.

Can equilibrium prices be computed with small amount of information even at the
microscopic level? This quest is clearly hopeless for general valuation functions.
But for one of the most important classes of valuation functions in economics,
\emph{gross substitute valuations}, there are enough structures to allow us to
construct equilibrium prices using microscopic information.

The history of \emph{gross substitutes} is intertwined with the development of
theory of Walrasian equilibrium (also called market equilibrium in this paper). Indeed, Kelso and Crawford \cite{KelsoCrawford}
show that Walrasian equilibrium always exists for gross substitute valuations.
Hatfield and Milgrom \cite{HatfieldMilgrom} argue that most important examples of valuation
functions arising in matching markets belong to the class of gross substitutes.
Gross substitutes have been used to guarantee the convergence of salary
adjustment processes in the job market \cite{KelsoCrawford},
to guarantee the existence of stable matchings
in a variety of settings \cite{Roth84, Kojima}, to show the stability of trading
networks \cite{HatfieldKNOW15}, to design combinatorial auctions
\cite{AusubelMilgrom,Shioura13} and even to analyze settings with
complementarities \cite{SunYang2006,HatfieldK15}.

Since the oracle access is very local, we clearly need to query each agent
at least once, so the dependence on the number of buyers needs to be at least
linear. We show that indeed it is possible to solve this problem with a linear
dependence on the number of buyers and cubic dependence in the number of items:

\begin{theorem*}[informal]
In a consumer market with $n$ goods and $m$ buyers whose valuation
functions satisfy the gross substitute condition, we can find an equilibrium (or Walrasian)
price and allocation using $mn+\widetilde{O}(n^{3})$ calls to the value oracle and
$\widetilde{O}(n^{3})$ time.
\end{theorem*}

Now that we have buyer-specific information, we can also compute the optimal
allocation at no additional time.

Proving this result requires novel insights into the structure of gross
substitute valuations. In particular, one of our main structural lemmas answers
a question posed by Hsu et al. \cite{hsu_prices}: \emph{when do prices coordinate
markets?}  In general, Walrasian prices mean that there is a choice of 
favorite bundles for each buyer that clears the market. It is far from trivial how to choose
those bundles, since each agent can have multiple favorite bundles (for
example, consider the case where all items are identical and all agents have the
same valuation for the items).
We say that a price vector form \emph{robust Walrasian prices} if
each agent demands a unique bundle under those prices and the bundles clear the
market. Such vectors would allow prices alone to clear the market without any
external coordination. We show that:

\begin{theorem*}[informal]
In a consumer market with $n$ goods and $m$ buyers whose valuation
functions satisfy the gross substitute condition, robust Walrasian
prices exist if and only if there is a unique Walrasian allocation. Whenever they
exist, they can be computed in $\widetilde{O}(mn + n^3)$ time.
\end{theorem*}

\subsection{Our algorithms and techniques}

We study the Walrasian equilibrium problem in three different settings: (i)
general valuations in the aggregate demand oracle model; (ii) gross substitute
valuations in the aggregate demand oracle model and (iii) gross substitutes in
the value oracle model. In all three settings, the starting point is the linear
programming formulation of Bikhchandani and Mamer \cite{bikchandani_mamer}.

\paragraph{General valuations in the aggregate demand oracle model (Section 3).} The main
difficulty of working with the LP in Bikhchandani and Mamer
\cite{bikchandani_mamer} is that the constraints depend on the value of buyers
for each bundle, which we have no access to. In particular, we are not able to
test for any given setting of variables of the LP if it is feasible or not. We
are in a strange situation where the LP to be solved is not known in full. Our solution is to move the
constraints to the objective function and turn the problem into an unconstrained
convex minimization problem. The problem in hand has the feature that we cannot
evaluate the function but we can compute its subgradients.

Traditional cutting plane algorithms such as the Ellipsoid Method need access to
both a separation oracle and functional values. To overcome this issue we use
the fact that the cutting plane of Lee, Sidford and Wong \cite{LeeSW15} provides
strong dual guarantees and that our separation oracle is given by subgradients.
Originally, in Theorem 42 of their paper \cite{LeeSW15} show how to adapt their
cutting plane method to convex optimization;
however, their algorithm and proof still rely on being able to evaluate the
function. We show in Theorem \ref{thm:convex_opt} that their algorithm can
be slightly modified to achieve the same guarantee using only subgradients
(i.e., without using functional values).

A second obstacle we face is that algorithms to minimize convex functions only
provide approximate guarantees and to find a Walrasian equilibrium we need the
exact minimum. In general minimizing a convex function exactly is impossible,
but in our case this can be done by exploiting the connection to the LP.
Note that given the very restricted way that we can access the problem (only via
the aggregate demand oracle), we cannot apply the Khachiyan's perturbation and rounding techniques for linear programming \cite{khachiyan1980polynomial} in a black-box fashion. Nevertheless his approach
 can be nontrivially adapted to our setting. We show how to perturb and
round the objective function to achieve the desired running time.

\paragraph{Gross substitutes in the aggregate demand model (Section 4).} If valuation
functions satisfy gross substitutes, then we can exploit the fact that the set
of Walrasian prices form an integral polytope with a lattice structure to
simplify the algorithm and obtain an improved running time and oracle call
complexity. The improvement comes from using structural properties of gross
substitutes to show that a simpler and milder perturbation to the objective is
enough and that rounding can be done in a simpler manner. This highlights the
important of looking at perturbation and rounding in a non-black-box manner.

\paragraph{Gross substitutes in the value oracle model (Section 5).} An aggregate demand
oracle call can be simulated from $O(mn^2)$ value oracles calls. This can be
plugged into the previous algorithm to obtain a running time of $\tilde{O}(mn^3
T_V)$ where $T_V$ is the time required by the value oracle. We use two ideas to 
improve the running time to $\tilde{O}((mn + n^3) T_{V})$. The first one is
to regularize the objective function. As with the use of regularizers in other context
in optimization, this is to penalize the algorithm for being too aggressive.
The bound of $O(mn^{2})$ value oracle calls per iteration of the cutting plane
algorithm is so costly precisely because we are trying to take
an aggressively large step. A second idea is
to re-use one of the stages of the subgradient computation
in multiple iterations, amortizing its cost per iteration.

\paragraph{Robust Walrasian Prices and Market Coordination (Section 6).}
Still in the value oracle model, we show how to obtain the efficient allocation from
the subgradients observed in the optimization procedure. An important by-product
of our analysis is that we give necessary and sufficient conditions for the
existence of robust Walrasian prices, i.e., Walrasian prices under which each
buyer has an unique bundle in their demand set. Whenever such prices exist, we
give an $\tilde{O}(mn + n^3)$ algorithm to compute them.
This answers an open question in Hsu et al \cite{hsu_prices}, who ask when it
is possible to completely coordinate markets by solely using prices.

\paragraph{Combinatorial Algorithms for Walrasian equilibrium (Section 7).}
Murota and Tamura \cite{murota_tamura} give combinatorial algorithms for the
problem of computing Walrasian equilibria via a reduction to the $M$-convex
submodular flow problem. It is also possible to obtain combinatorial algorithms
for the welfare problem by reducing it to the  valuated matroid intersection
problem and applying the algorithms in Murota \cite{Murota96a,Murota96b}.
The running time is not explicitly analyzed in
\cite{murota_tamura,Murota96a,Murota96b}. Here we describe those algorithms for
the reader unfamiliar with M-convex submodular flows in
terms of a sequence of elementary shortest path computations and analyze its
running time. We show that they have running time of $\tilde{O}(mn^3)$ in the
value oracle model. We use the same ideas used to speed up the computation of
Walrasian prices by regularizing the market potential to speed up those
algorithms and improve its running time to $\tilde{O}(mn + n^3)$.

\subsection{Comparison to related work}

\paragraph{Iterative auctions and subgradient algorithms.}
The first algorithm for computing Walrasian equilibria in an economy of
indivisible goods is due to Kelso and Crawford \cite{KelsoCrawford} and it is
inspired by Walras' t\^{a}tonnement procedure \cite{walras1874}, which means
``trial-and-error''. Despite the name, it constitutes a very ingenious greedy
algorithm: goods start with any price, then we compute the aggregate demand of
the agents, increase the price by one for all goods that were over-demanded and
decrease by one the price of all goods that are under-demanded. This gives a
very natural and simple algorithm in the aggregate demand oracle model. This
algorithm, however, is not polynomial time since it runs in time proportional to
the magnitude of the valuations.

The seminal work of \cite{KelsoCrawford} originated two lines of attack of the
problem of computing Walrasian equilibria: the first line is by applying
subgradient descent methods 
\cite{parkes1999bundle,parkes2002ascending,AusubelMilgrom}.
Such methods typically either only guarantee convergence to an approximate
solution or converge in pseudo-polynomial time to an exact solution.
This is unavoidable if subgradient descent methods are used, since their
convergence guarantee is polynomial in $M/\epsilon$ where $M$ is the maximum
valuation of a buyer for a bundle and $\epsilon$ is the accuracy of the
desired solution. To be able to round to an \emph{exact} optimal solution
the running time must depend on the magnitude of the valuations.
Another family of methods is based on primal-dual
algorithms. We refer to de Vries, Schummer and Vohra \cite{de2007ascending} for
a systematic treatment of the topic. For primal-dual methods to converge exactly, they
need to update the prices slowly -- in fact, in \cite{de2007ascending} prices
change by one unit in each iteration -- causing the running time to be
pseudo-polynomial time.

\paragraph{Polynomial time approaches via the Ellipsoid Method.}
The Welfare Problem for gross substitutes was independently shown to be solvable
in polynomial by Murota \cite{Murota96b} and Nisan and Segal
\cite{NisanSegal}. Remarkably, this was done using completely different methods.

Nisan and Segal's approach is based on a linear programming
formulation of Walrasian equilibrium due to Bikhchandani and Mamer
\cite{bikchandani_mamer}. The authors show that the dual of this formulation can
be solved using both the value and demand oracles for gross substitutes as
a separation oracle for the LP. This can be combined with the
fact that demand oracles for gross substitutes can be constructed from value
oracles in $O(n^2)$ time \cite{DressTerhalle_WellLayered} to obtain a
polynomial-time algorithm in value oracle model.

This is the method that is closer to ours in spirit: since we both approach the
problem via a mathematical programming formulation and apply interior point
methods. In terms of oracle access,
Nisan and Segal crucially rely on value oracles to implement the separation
oracle in their LP -- so their solution wouldn't generalize to the
aggregate demand oracle model, since neither per-agent demand
oracles nor value oracles can be recovered from aggregate demand
oracle\footnote{Note that the construction of Blumrosen and Nisan
\cite{blumrosen} to construct value oracles from demand oracles crucially
requires \emph{per-buyer} demand oracles. The same construction doesn't
carry over to aggregate demand oracles.}.
The running time in their paper is never formally
analyzed, but since their formulation has $m+n$ variables, it would lead to a
superlinear dependence in the number of agents.

Nisan and Segal employ the LP to compute a set of Walrasian prices
and the value of the Walrasian allocation. In order to compute the
allocation itself, they employ a clever technique called
\emph{self-reducibility}, commonly used in complexity theory. While it allows
for an elegant argument, it is a very inefficient technique, since it requires
solving $nm$ linear programs. In total, this would lead to a running time of
$O(mn^2(m+n)^3)$ using currently fastest cutting plane algorithms as the LP solver.

\paragraph{Combinatorial approaches.}
A second technique was developed by Murota \cite{Murota96a,Murota96b} and leads
to very efficient combinatorial algorithms. Murota's original paper never
mentions the term ``gross substitutes''. They were developed
having a different object in mind, called \emph{valuated matroids}, introduced
by Dress and Wenzel \cite{DressWenzel92,DressWenzel} as a generalization of the
Grassmann-Pl\"{u}cker relations in $p$-adic analysis. Murota developed a
strongly-polynomial time algorithm based on network flows for a problem called 
the \emph{valuated matroids assignment problem}. There is a tortuous path
connecting gross substitutes to valuated matroids. Valuated matroid turned out
to be one aspect of a larger theory, Discrete Convex Analysis, developed by Murota (see his book
\cite{murota2003discretebook} for a comprehensive discussion). One central
object of this theory is the concept of $M^\natural$-concave
functions, introduced by Murota and Shioura \cite{MurotaShioura}. It came to
many as a surprise when Fujishige and Yang \cite{FujishigeYang} showed that
$M^\natural$-concave functions and gross substitutes are the same thing. Their
equivalence is highly non-trivial and their definitions are very different to
the point it took at least a decade for those concepts to be connected. Murota
and Tamura \cite{murota_tamura} later apply the ideas in discrete convex analysis
to give polynomial time algorithms to various equilibrium problems in economics.
The running time is never explicitly analyzed in their paper. Here we show that
their running time is $\tilde{O}(mn^3)$ and improve it to $\tilde{O}(mn + n^3)$

\paragraph{Market Coordination.} Related to Section
\ref{sec:robust_walrasian_prices} in our paper is the line of research on Market
Coordination. This line of inquiry was initiated by Hsu, Morgenstern, Rogers,
Roth and Vohra \cite{hsu_prices} who pose the question of when prices are enough
to coordinate markets. More
precisely, they showed that under some genericity condition the minimal
Walrasian price \textit{for a restricted class of gross substitutes} induces an
overdemand at most $1$ for each item. On the other hand, we show in
Theorem \ref{thm:robust_walrasian} that
under a more inclusive condition (which is necessary and sufficient)
\textit{almost every} Walrasian
prices \textit{for any gross substitutes} have \textit{no} overdemand, i.e. the
market is perfectly coordinated.  We also give a simple algorithm for
computing those prices whenever they exist. Such necessary and sufficient
conditions were given simultaneously and independently by Cohen-Addad, Eden,
Feldman and Fiat \cite{AddadEFF15}\footnote{Both papers were submitted to arxiv
in the same week.}.

\subsection{Conclusion and Discussion}

We provide in this paper the first polynomial-time algorithm for computing
Walrasian prices with an aggregate demand oracle. Previous algorithms for this
problem required either a value oracle or
a per-buyer demand oracle, or both. We also gave the fastest (to the best of our
knowledge) algorithm for computing a Walrasian equilibrium in economies with
gross substitute valuations. En route, we showed necessary and sufficient
conditions for the existence of robust Walrasian prices and provided an
algorithm to compute them.

We also provide the fastest known algorithm for Walrasian equilibrium in the
value oracle model. We believe the question of improving the running time is
important because it leads to new algorithmic ideas and new structural
insights. For example, the question of the existence of robust Walrasian prices
(Theorem \ref{thm:robust_walrasian}) arises as a step towards computing the
Walrasian allocation from subgradients. Later we noticed this structural lemma
also provided answer to a purely economic question of independent interest. A
second example is that by seeking to improve the running time, we looked for
algorithms that would perform very simple operations and this made us stumble
upon algorithms that used only the aggregate demand oracle.

One of the salient features of our first result is that we compute Walrasian
prices whenever they exist for any class of valuations. The reader might ask why
this is interesting since Gul and Stachetti \cite{GulStachetti} show that gross
substitutes are the largest class of valuations for which Walrasian equilibria
exist. The confusion stems from the qualification in their result. What they show
is if a class of valuation function is closed under summing an additive functions
and for all vectors of valuations in this class there is Walrasian equilibrium,
then it is contained in the class of gross substitutes. Such statement
doesn't preclude existence of equilibrium for more general classes. For example,
if there is a single buyer, then equilibrium exists no matter what his valuation
is. Also, rich classes of valuation functions were shown to always have
Walrasian equilibrium \cite{ozan14,ozan15,BenZwi,SunYang2006}. Of course, none
of those classes are closed under summing an additive function.

\section{Preliminaries}

\subsection{Notation}
Let $\Z$ be the set of integers, $\Z_{\geq 0}$ be the set of non-negative
integers. For any $k \in \Z_{\geq 0}$, we define $[k] := \{1, 2, \hdots, k\}$
and $\dbracket{k} = \{0, 1, \hdots, k\}$.
Also, give a vector $\vc{s} = (s_1, \hdots, s_n) \in \Z_{\geq 0}^n$, we define
$\dbracket{\vc{s}} = \prod_{j=1}^n \dbracket{s_i}$.

\subsection{Market equilibrium: prices, allocations and welfare}

Following the classic model of Gul and Stachetti \cite{GulStachetti}, we define
a \emph{market of indivisible goods} as a set $[m]$ of buyers, a set $[n]$ of
items and a supply $s_j\in \Z_{\geq 0}$ of each item $j$.
Each buyer $i$ has a valuation function
 $v_i : \Z_{\geq 0}^n \rightarrow \Z$
over the multisets of items with $v_i(0) = 0$. For the first part of our paper, we make no
further assumptions about the valuation function or the supply.

Given a price vector $\vc{p}\in\R^{n}$, we define the \emph{utility} of agent $i$ for a
bundle $x \in \dbracket{\vc{s}}$ under price vector $\vc{p}$ as:
$u_i(x;\vc{p}) := v_{i}(x)-\vc{p}\cdot x$, where $\vc{p} \cdot x$ refers to the
standard dot product $\sum_{j=1}^n p_i x_i$. Notice also that we make no
assumptions about the signs of $v_i$ and $\vc{p}$.

An \emph{allocation} $\vc{x}=(x^{(1)},x^{(2)},\ldots,x^{(m)})$
is simply an assignment of items to each player where $i$ receives
$x^{(i)}\in\dbracket{\vc{s}}$. An allocation $\vc{x}$ is \emph{valid} if
it forms a partition of the items, i.e. $\sum_{i}x_{j}^{(i)}=s_{j}$ for every
$j$. The \emph{social welfare} of a valid allocation $\vc{x}$ is defined
as $\sw(\vc{x}) = \sum_{i\in[m]}v_{i}(x^{(i)})$. Finally, the optimal social welfare
is simply the largest possible social welfare among all valid allocations.

In the following, we show the importance of prices in welfare economics,
namely that if the market clears, then we achieve the optimal social
welfare.

Given prices $\vc{p}\in\R^{n}$, we would expect a rational agent
to buy $x$ such that his utility $v_{i}(x)-\vc{p}\cdot x$ is maximized.
We call $x$ the demand of $i$ under $p$, as defined formally
below. Note that there may be multiple utility-maximizing subsets.
\begin{definition}[Demand set]
 Given prices $\vc{p}\in\R^{n}$ on each item, the demand set
$D(v,\vc{p})$ for a valuation function $v$ is the collection of subsets
for which the utility is maximized:
$$
D(v,\vc{p}):=\arg\max_{x\in\dbracket{\vc{s}}}v(x)-\vc{p}\cdot x.
$$
If $v$ is the valuation function $v_{i}$ of player $i$, we also
use the shorthand $D(i,\vc{p})$ as the demand set.
\end{definition}
We are now ready to define competitive equilibrium.
\begin{definition}[Equilibrium]
A \emph{Walrasian equilibrium} (also called \emph{competitive equilibrium})
consists of a price vector $\vc{p}\in\R^{n}$ and a valid allocation
$\vc{x}=(x^{(1)},x^{(2)},\ldots,x^{(m)})$
such that $x^{(i)}\in D(i,\vc{p})$ for all $i$. We call $\vc{p}$ an \emph{equilibrium/Walrasian
price} and $\vc{x}$ an \emph{equilibrium/Walrasian allocation} induced by $\vc{p}$.
\end{definition}
In other words, a competitive equilibrium describes a situation where
items are sold in such a way that the total demand $\sum_{i}x_{j}^{(i)}$
for each item precisely meets its supply $s_{j}$, i.e. the market
\emph{clears}. The reason for the name \emph{competitive} equilibrium
is that its achieve the optimal social welfare. This is known as the
first and second welfare theorems in economics and for completeness we provide
a proof in the appendix.

\begin{lemma}[First and Second Welfare Theorems]\label{lemma:welfare_thm}
Let $\vc{x}=(x^{(1)},x^{(2)},\ldots,x^{(m)})$ be an equilibrium allocation
induced by an equilibrium prices $\vc{p}$. Then $\vc{x}$ achieves the optimal
social welfare.

Moreover, if $\vc{p}$ is any set of Walrasian prices and $\vc{x}$ is any optimal
allocation, then the pair $(\vc{p}, \vc{x})$ form a Walrasian equilibrium.
\end{lemma}
This lemma nicely reduces social welfare maximization to finding competitive
equilibrium whenever it exists. Note that the definition of social
welfare has nothing to do with prices. In a way, this lemma shows
that equilibrium prices act as a certificate which demonstrates the
optimality of equilibrium allocation.

\subsection{Oracles}

To study market equilibrium computationally, we must clarify how the
market is represented and how we can extract information about it.
In this paper we consider three models that access the market in different
scales:

\subsubsection{Microscopic Scale: Value Oracle Model}

In the \emph{value oracle model} the algorithm has access to the value that each
agent has for each bundle. This gives the algorithm very fine-grained
information, but it requires potentially many calls for the algorithm to access any
sort of macroscopic information about the market.

\begin{definition}[Value oracle]
The value oracle for player $i\in[m]$ takes $x\in\dbracket{\vc{s}}$
as input and outputs $v_{i}(x)$. We denote by $T_{V}$ the time spent
by the value oracle to answer a query.
\end{definition}

\subsubsection{Agent Scale: Demand Oracle Model}

In the \emph{demand oracle model} the algorithm presents a price
vector $\vc{p}$ to an agent and obtains how many units are demanded at that
price. At any given price, the agent could be indifferent between various
bundles. The demand oracle can return any arbitrary bundle.

\begin{definition}[Demand oracle]
The demand oracle for valuation agent $i$ takes as input a price vector
$\vc{p}\in\R^{n}$ and outputs a demand vector $d_i(\vc{p}) \in D(i,\vc{p})$.
We denote by $T_{D}$ the time spent by the demand oracle to answer 
a query.
\end{definition}

\subsubsection{Macroscopic (Market) Scale: Aggregate Demand Oracle Model}

The \emph{aggregate demand oracle model} presents a very macroscopic view of the
market. In this model, the algorithm cannot observe individual agents but only
the aggregate response of the market to any given price $\vc{p}$. For example, a
manufacturer deploying a product in the market is unable to observe each buyer's behavior, but
only how many units were sold.

\begin{definition}[Aggregate Demand oracle]
The aggregate demand oracle takes as input a price vector
$\vc{p}\in\R^{n}$ and outputs a demand vector $d(\vc{p}) \in \Z_{\geq 0}$ such
that there exist bundles $x^{(i)} \in D(i,\vc{p})$ satisfying $d(\vc{p}) = \sum_i
x^{(i)}$. We denote by $T_{AD}$ the time spent by the demand oracle to answer 
a query.
\end{definition}

\subsection{Convex analysis}

As we tackle the problem of finding market equilibrium using convex minimization,
we need a few basic facts about subgradients~\cite{rockafellar2015convex}. For differentiable functions, subgradients are just
gradients. Throughout this section all functions are continuous, convex,
real-valued, and defined over a convex subset of $\R^{n}$.
\begin{definition}
Let $f$ be a convex function on $\R^{n}$. $\vc{g}$ is a subgradient
of $f$ at $\vc{p}'$ if for any other $\vc{p}$ (in the domain),
$f(\vc{p})-f(\vc{p}')\geq \vc{g}\cdot(\vc{p}-\vc{p}').$
The set of subgradients at $\vc{p}$ is denoted by $\partial f(\vc{p})$. Sometimes
we abuse notation by denoting a subgradient simply as $\partial f(\vc{p})$.
\end{definition}
It is well-known that every continuous convex function has a subgradient
everywhere. Subgradients are nice particularly because they provide
a ``separation oracle'' in the following sense. Note that this is
almost a tautology.
\begin{lemma}
Let $\vc{g}$ be a subgradient at $\vc{p}'$ of convex $f$. If $\vc{p}$ minimizes
$f$, then 
$\vc{g}\cdot(\vc{p}'-\vc{p})\leq0$.

\end{lemma}
By using $\vc{g}\cdot(\vc{p}'-\vc{p})\leq0$ as a separating hyperplane, this is
the basis on which subgradients allow us to solve convex minimization
via the ellipsoid or cutting plane method in polynomial time.

We can identify a minimizer by looking at its subgradients.
\begin{lemma}
For convex $f$, $\vc{p}$ minimizes $f$ iff $0$ is a subgradient at $\vc{p}$.
\end{lemma}
Finally, we will frequently take the subgradient of a function in
the form $h(\vc{p})=\max_{i\in I}h_{i}(\vc{p})$. 
\begin{theorem}[Envelope Theorem]\label{thm:envelope}
Let $h(\vc{p})=\max_{i\in I}h_{i}(\vc{p})$ where $I$ is an index set and $h_{i}(\vc{p})$'s
are all convex. Then $h$ is convex and $\partial h(\vc{p})$ is the convex
hull of the subgradient $\partial h_{i}(\vc{p})$ $for$ $i\in\arg\max_{i\in I}h_{i}(\vc{p})$.
In particular, any subgradient of $h_{i}(\vc{p})$ is a subgradient of
$h(\vc{p})$ whenever $i\in\arg\max_{i\in I}h_{i}(\vc{p})$.\end{theorem}

\section{Walrasian equilibrium for General Valuations in $\tilde{O}(n^2 \cdot
T_{AD} + n^5)$}\label{sec:general}

We show that in the aggregate demand oracle model, whenever a Walrasian
equilibrium exists, it can be computed using $\tilde{O}(n^2)$ aggregate demand
oracles calls and $\tilde{O}(n^2 \cdot T_{AD} + n^5)$ time.
We want to emphasize that our result is
in the \emph{aggregate} demand oracle model -- which is the typical information
available to markets: \emph{which goods are under- and over-demanded by the population
of buyers?} Previous polynomial-time algorithms for computing Walrasian equilibria
\cite{bikchandani_mamer,NisanSegal,Murota96b,murota_tamura} require buyer-level demand oracles
or value oracles.
Previous algorithms that use only aggregate demand oracles (such as
\cite{KelsoCrawford,GulStachetti,AusubelMilgrom,parkes2002ascending,de2007ascending}
and related methods based on ascending
auctions) are pseudopolynomial since they depend linearly on the
magnitude of the valuation functions.

First we discuss a linear programming formulation for this problem and a
corresponding convex programming formulation. The formulation itself is fairly
standard and appears in various previous work. When we try to find an
\emph{exact} minimizer of this formulation in polynomial time using only
aggregate demand oracles, we encounter a series of obstacles. The main ones are:
(i) how to optimize a function we cannot evaluate; and (ii) how to find an exact
minimizer using convex optimization. In both cases, we need to apply
optimization algorithms in a \emph{non-black-box} manner and exploit special
structure of the problem.

\paragraph{Linear Programming Formulation} We start from the
formulation of Bikhchandani and Mamer \cite{bikchandani_mamer} (also studied by
Nisan and Segal \cite{NisanSegal}) of the Walrasian equilibrium problem as a
linear program. Consider the following primal-dual pair:

$$\left.\begin{aligned} & \max_z \sum_{i \in [m],x \in \dbracket{\vc{s}}} v_i(x) \cdot z_{i,x} \text{ s.t. }\\
  & \quad \begin{aligned} & \sum_{x} z_{i,x} = 1, & \forall i \quad & (u_i) \\
  & \sum_{i, x} x_j \cdot z_{i,x} = s_j, & \forall j \quad & (p_j) \\
  & z_{i,x} \geq 0, & \forall i,x \end{aligned} \end{aligned} \qquad \right| \qquad
  \begin{aligned} & \min_{(\vc{p},\vc{u})} \sum_i u_i + \vc{p} \cdot \vc{s} \text{ s.t. } \\
    & \quad \begin{aligned}
      & u_i \geq v_i(x) - \vc{p} \cdot x, & & \forall i,x & \quad (z_{i,x}) 
    \end{aligned} \end{aligned} $$

    \begin{lemma}[\cite{bikchandani_mamer}]
  \label{lem:eqismin}If a market equilibrium $(\vc{p}^{\eql},\vc{x})$ exists iff
      the primal program has an integral solution. In such case,
      the set of equilibrium prices is the set of solutions to the dual LP
      projected to the $\vc{p}$-coordinate.\end{lemma}

We provide a proof of the previous lemma in the appendix for completeness.
This lemma reduces the problem of finding a Walrasian equilibrium,
whenever it exists, to the problem of solving the dual program.
The approach in Nisan and Segal \cite{NisanSegal} is to use a (per buyer)
demand oracle as a separation oracle for the dual program. Since we only have access
to an aggregate demand oracle, we will consider a slightly different LP: given that
we care only about the $\vc{p}$ variables, we can reformulate the dual as:

    \begin{equation}\label{dual_program}\tag{D}
    \begin{aligned} & \min_{(\vc{p}, u)} u + \vc{p} \cdot \vc{s} \text{ s.t. } \\
    & \quad \begin{aligned}
      & u \geq \sum_i v_i(x^{(i)}) - \vc{p} \cdot x^{(i)}, & & \forall x_i \in
      \dbracket{\vc{s}}  \end{aligned} \end{aligned}
    \end{equation}

It is simple to see that for every feasible vector $(\vc{p}, \vc{u})$ of the
  original dual, we can find a corresponding point $(\vc{p}, \sum_i u_i)$ of the
  transformed dual with the same value. Conversely, given a feasible point
  $(\vc{p}, u)$ of the transformed dual, we can come up with a point $(\vc{p},
  \vc{u})$ of the original dual with equal or better value by setting $u_i =
  \max_x v_i(x) - \vc{p} \cdot x$.

Thus it suffices to find an optimal solution to the transformed dual
  program. The separation problem is now simpler: consider a point
  $(\vc{p_0}, u_0)$ that is infeasible. If some constraint is violated, then it must
  be the constraint for $x^{(i)}$ maximizing $\sum_i v_i(x^{(i)}) - \vc{p}_0 \cdot
  x^{(i)}$, so $u_0 < \sum_i v_i(x^{(i)}) -  \vc{p}_0 \cdot x^{(i)}$. For all
  feasible $(\vc{p}, u)$ we know that $u \geq \sum_i  v_i(x^{(i)}) - \vc{p}
  \cdot x^{(i)}$. Therefore: $$u - u_0 + (\vc{p} - \vc{p}_0) \cdot
  \vc{d}(\vc{p}_0) \geq 0$$
  is a valid separating hyperplane, where $\vc{d}(\vc{p}_0) = \sum_i x^{(i)}$ is the
  output of the aggregate demand oracle. If on the other hand $(\vc{p_0}, u_0)$
  is feasible, we can use the objective function to find a separating
  hyperplane, since for any optimal solution $(\vc{p}, u)$ we know
  that $u + \vc{p} \cdot \vc{s} \leq u_0 + \vc{p}_0 \cdot \vc{s}$. So we can
  separate it using: $$u - u_0 + (\vc{p} - \vc{p}_0) \cdot
  \vc{s} \leq 0$$

  \paragraph{An Obstacle} The main obstacle to this approach is that since the
  aggregate demand oracle has no access to the value of $\sum_i v_i(x^{(i)})$,
  it offers no information of whether a vector $(\vc{p}_0, u_0)$ is feasible or
  not, and so it is not clear which separation hyperplane to use.

  \paragraph{Convex Programming Formulation} One way to get around this obstacle
  is to further transform the dual program to get rid of utility variable $u$
  altogether. We do so by moving the constraints to the objective function in
  the form of penalties. The problem then becomes the problem of minimizing a
  convex function. The same function has appeared in \cite{ausubel2006efficient}
  as a potential function used to analyze a differential equation governing a
  price adjustment procedure and in \cite{Shioura13} as a potential function to
  measure the progress of a combinatorial algorithm.

  Given a market with supply $\vc{s}$ and agents with valuations $v_1,
  \hdots, v_m$, we define the market potential function $f:\R^n \rightarrow \R$ as:

  \begin{equation}\label{convex_f}\tag{C}
  f(\vc{p})=\sum_{i=1}^{m}\left(\max_{x\in\dbracket{\vc{s}}}v_{i}(x)-\vc{p}\cdot
  x\right)+\vc{p}\cdot \vc{s}.
  \end{equation}

  Since $f$ is nothing more than the dual linear program with the constraints
  moved to the objective function, the set of minimizers of $f$ is exactly the
  set of optimal solutions of the dual linear program (or more precisely, their
  projections to the $\vc{p}$-variable). Each term
  $\max_{x\in\dbracket{\vc{s}}}v_{i}(x)-\vc{p}\cdot x$ is a convex function in
  $\vc{p}$ since it is a maximum over linear functions. Hence $f$
  is convex since it is a sum of convex functions.

  One remarkable fact about $f$ is that we can obtain a subgradient from the
  aggregate demand oracle:

\begin{lemma}[Subgradient oracle]
  \label{lem:subgradient} Let $\vc{d}(\vc{p})$ be the aggregate demand,
  then $\vc{s} - \vc{d}(\vc{p})$ is a subgradient of $f$ in $\vc{p}$.
\end{lemma}

\begin{proof}
  From the Envelope Theorem (Theorem \ref{thm:envelope}), a subgradient of
  $\max_{x\in\dbracket{\vc{s}}}v_{i}(x)-\vc{p}\cdot x$ is the subgradient of
  $v_i(x^{(i)}) - \vc{p} \cdot x^{(i)}$ for $x^{(i)} \in \arg \max_x
  v_{i}(x)-\vc{p}\cdot x$. Since $v_i(x^{(i)}) - \vc{p} \cdot x^{(i)}$ is a
  linear function in $\vc{p}$, its gradient is simply $-x^{(i)}$. So, for any
  $x^{(i)}$ in the $\arg \max$, the vector $\vc{s} -  \sum_i x^{(i)}$ is a
  subgradient o f $f$. In particular, $\vc{s} - \vc{d}(\vc{p})$.
\end{proof}

A useful fact in studying this function is that we have an initial bound on the
set of minimizers:

\begin{lemma}\label{lemma:initial_box} If there exist a Walrasian equilibrium,
then the set of minimizers $P := \arg \min_{\vc{p}} f(\vc{p})$ is
contained in the box $[-2M, 2M]^n$ for $M = \max_i \max_{x \in \dbracket{s}}
\abs{v_i(x)}$.
\end{lemma}

\begin{proof}
Let $\vc{p} \in P$ be a vector of equilibrium prices and
$\vc{x}$ an optimal allocation. Since the pair $(\vc{p}, \vc{x})$ constitute a
Walrasian equilibrium, then for any item $j$, there is
some $x_{j}^{(i)}\geq1$ and therefore,
$$
v_{i}(x^{(i)})-\vc{p}\cdot x^{(i)}\geq
v_{i}(x^{(i)}-\vc{1}_{j})-\vc{p}\cdot(x^{(i)}-\vc{1}_{j})
$$
where $\vc{1}_j$ is the unit vector in the $j$-th coordinate.
This gives us: $p_{j}\leq v_{i}(x^{(i)})-v_{i}(x^{(i)}-\vc{1}_{j})\leq2M$.

For the upper bound, If there is more than one buyer then
$\vc{p}_j$ must be larger than $-2M$, otherwise all the supply of item
$j$ will be demanded by all buyers and therefore $\vc{p}$ can't be Walrasian.
\end{proof}

\begin{lemma}\label{lemma:vertices}
The set of Walrasian prices is a polytope $P$ whose vertices have coordinates of the
form $p_i = a_i / b_i$ with $a_i, b_i \in \Z$ and $\abs{b_i} \leq (Sn)^n$
for $S = \max_i s_i$.
\end{lemma}
\begin{proof}
The vertices of $P$ correspond to
$\vc{p}$-coordinates of the basic feasible solutions of the dual linear program.
Given that the coefficients of
the $\vc{p}$ variables in the linear program are integers from $0$ to $S$.
Solving the linear system using Cramer's rule (see Section 5 in
\cite{bland1981ellipsoid}) we get that every solution must be a fraction with
denominator at most $n! \cdot S^n \leq (Sn)^n$.
\end{proof}

\paragraph{Two More Obstacles} The natural approach at this point is to try to
apply convex optimization algorithms as a black box to optimize $f$. There are
two issues with this approach. The first, which is easier to address, is that
unlike algorithms for linear programming which are exact, algorithms for general
convex optimization methods give only $\epsilon$-approximation guarantees.
We will get around this issue by exploiting the connection between $f$ and the
linear program from which it is derived.

A second more serious obstacle is the fact that we don't have access to the
functional value of $f$, only to its subgradient. This is a problem for the
following reasons. The way that traditional cutting plane methods work is that
they keep in each iteration $t$ a set $G_t$ called a \emph{localizer}. The
localizer is a subset of the domain which is guarantee to contain the optimal
solution. We start with a large enough localizer set $G_0$ that is guaranteed to
contain the solution. In each iteration $t$, a point $\vc{p}_t \in G_{t-1}$
is queried for the subgradient $\partial f(\vc{p}_t)$. Now, if $\vc{p}^*$ is an
optimal solution we know that:
$$0 \geq f(\vc{p}^*) - f(\vc{p}_t) \geq \partial f(\vc{p}_t) \cdot (\vc{p}^* -
\vc{p}_t)$$
where the first inequality comes from the fact $\vc{p}^*$ is an optimal solution
and the second inequality comes from the definition of the subgradient. This in
particular means that any optimal solution must be contained in $H_t = G_{t-1} \cap
\{\vc{p};  f(\vc{p}_t) \cdot (\vc{p} -\vc{p}_t) \leq 0 \}$. The method then
updates $G_t$ to either $H_t$ or a superset thereof. The Ellipsoid Methods,
for example, updates $G_t$ to the smallest ellipsoid containing $H_t$. So far,
all the steps depend only on the subgradient and not on the actual functional
values of $f(\vc{p}_t)$. The guarantee of the method, however, is that after a
sufficient number of steps, one of the iterates is close to the optimal, i.e.,
$\min_t f(\vc{p}_t) - f(\vc{p}^*) \leq \epsilon$. To find the approximate
minimizer, however, we need to look at the actual functional values, which
we have not access to. See Section 2.2 in Nemirovski's book
\cite{nemirovski2005efficient} for a complete discussion on the guarantees
provided by cutting plane methods.

\paragraph{A special case : Walrasian prices with non-empty interior}
The set of Walrasian prices $P$ forms a convex set, since it corresponds to the set
of minimizers of the convex function $f$. If $P$ has non-zero volume, it is
possible to find a Walrasian equilibrium using the Ellipsoid Method
without needing to know the functional value.
If the set is large, the Ellipsoid method can't avoid querying a point in
the interior of $P$ for too long. And when a point in the interior of $P$ is
queried, the only subgradient is zero, or in market terms, each agent has an
unique favorite bundle and those bundles clear the market. This means that the
aggregate demand oracle returns exactly the supply. Next we make this discussion
formal:

\begin{theorem}\label{thm:general_walrasian_non_empty}
Assume that there the set of Walrasian prices $P$ has non-zero
volume, then in $\tilde{O}(n^3)$ iterations of the Ellipsoid method\footnote{One may use cutting plane methods to derive a better running time but for simplicity we omit doing it for this special case.} is
guaranteed to query a Walrasian price $\vc{p}^*$ for which the aggregate demand
oracle returns $\vc{d}(\vc{p}^*) = \vc{s}$.
\end{theorem}

\begin{proof}
By Lemma \ref{lemma:initial_box}, $P \subseteq [-2M, 2M]^n$, so we can take the
initial set of localizers $G_0$ in the ellipsoid methods to be the ball of
radius $2M \sqrt{n}$ around $0$, which has volume $O(M)^n$. The Ellipsoid
guarantee (see Section 3 of \cite{nemirovski2005efficient}) is that the volume
of $G_t$ is at most $e^{-t/2n}$ of the initial volume. So if the method hasn't
queries any point in the interior of $P$, then we must have $G_t$ containing
$P$.

To bound the volume of $P$ we use the fact that is a polytope has vertices
$\vc{p}_0, \vc{p}_1, \hdots, \vc{p}_{n}$ then the volume of the polytope is
lower bounded by the volume of the convex hull of those vertices:
$$\vol(P) \geq \frac{1}{n!}\det \left[\begin{matrix}  1 & 1 & \hdots & 1 \\ \vc{p}_0 &
\vc{p}_1 & \hdots & \vc{p}_n  \end{matrix} \right] \geq \frac{1}{n!}
\left( \frac{1}{O(Sn)^{n}} \right)^n = \Omega((Sn)^{-n^2})$$
where the last inequality follows from Lemma \ref{lemma:vertices}.

Therefore, if $G_t$ contains $P$ then:
$O(M)^n e^{-t/2n} \geq \Omega((Sn)^{-n^2})$
which implies that $t \leq O(n^3 \log(Sn) + n^2 \log(M)) = \tilde{O}(n^3)$. So
after so many iterations we are guaranteed to query a point in the interior of
$P$. For a point $\vc{p}$ in the interior of $p$, there is a small ball around
it for which $f$ is flat, so zero must be the unique subgradient at that point.
Therefore, the aggregate demand oracle has no choice but to return
$\vc{d}(\vc{p}) = \vc{s}$.
\end{proof}

The main drawback of this method is that it strongly relies on
the fact that $P$ has non-empty interior. Some very simple and well-behaved
markets have a set of Walrasian prices of zero volume.
Consider for example a market with two items with
supply one of each and three buyers, each having valuation $v(0) = 0$
and $v(x) = 1$ otherwise. The set of Walrasian prices is $P=\{(1,1,1)\}$. Even
for $\vc{p}^* = (1,1,1)$, the subgradient is not unique since the aggregate
demand oracle can return any $\vc{d}(\vc{p}^*) = (d_1, d_2)$ for $0 \leq d_1 +
d_2 \leq 3$ and $d_1, d_2 \in \dbracket{3}$. In this case even trying to modify
the objective function is hopeless, since there is no point in the domain for
which the aggregate demand oracle is guaranteed to return the supply vector.
Since there is no point for which the subgradient oracle is guaranteed to return
zero, there is no direct way to recognize a vector of Walrasian prices even if
we see one!

\paragraph{Approach for the general case}
Our approach for optimizing $f$ is as follows: first we show how to obtain
an $\epsilon$-approximate minimizer of $f$ in time
$O(n  \log(nM / \epsilon) T_{AD} +  n^3\log^{O(1)}(nM / \epsilon))$.
In order to round the $\epsilon$-approximate solution to the optimal solution,
we exploit the fact that $f$ came from an linear program and customize the
traditional approach of Khachiyan \cite{khachiyan1980polynomial} for rounding
approximate to exact solutions.

The idea is to use Lemma \ref{lemma:vertices} to argue that if $f$ has a unique
minimizer, and we set $\epsilon$ small enough, then all
$\epsilon$-approximate solutions are in a small neighborhood of the exact optimal.
Moreover, for small enough $\epsilon$, there should be only one point in the
format indicated by Lemma \ref{lemma:vertices} in a small neighborhood of the
approximate minimizer, so we can recognize this as the exact solution by rounding.

For this approach to work, we need $f$ to have a unique minimizer. To achieve
that, we perturbe the objective function $f$ to $\hat{f}$ in such a way that
$\hat{f}$ has a unique minimizer that this minimizer is still a minimizer of
$f$. There are several ways to implement this approach, the simplest of which
uses the isolation lemma (which, by the way, was not available to Khachiyan at
the time so he had to resort to something more complicated).

\paragraph{A recent Cutting Plane Method without using functional values}
The first part of the algorithm consists in obtaining an $\epsilon$-approximate
minimizer of $f$ without using its functional value. In order to do so, we use a
recent technique introduced by Lee, Sidford and Wong \cite{LeeSW15}. The authors
show an efficient algorithm which either identifies a point inside the desired convex set $K$ or
\textit{certifies} that $K$ contains no ball of radius $\epsilon$. We show that
their main theorem (Theorem 31) can be used to compute
approximate minimizers of convex functions using only the subgradient. We note
that in their paper also provides an application of Theorem 31 to
minimizing convex functions (Theorem 42 in \cite{LeeSW15}) but their proof
relies on using functional values as they did not assume a separation oracle given by subgradients.

\begin{theorem}[Lee, Sidford, Wong (Theorem 31 in \cite{LeeSW15})]\label{thm:lsw}
  Let $K \subseteq [-M,M]^n$ be a convex set and a separation oracle can
  be queries for every point in $\vc{p} \in [-M, M]^n$ will either return that
  $\vc{p} \in K$ or return a half-space $H = \{\vc{p} \in \R^n; \vc{a}_i \cdot \vc{p}  \leq
  b_i\}$ containing $K$. Then there exists an algorithm with running time $O(n T
  \log(nM/\delta) + n^3 \log^{O(1)}(nM)\log^2(nM/\delta))$\footnote{The running time stated in their paper has $\log^{O(1)}(nM/\delta)$ dependence which is, upon a closer examination, actually $\log^{O(1)}(nM)\log^2(nM/\delta))$. (They ignored the difference because both runtimes give the same result for their applications)} that either produces a
  point $\vc{p} \in K$ or finds a polytope $P = \{\vc{p} \in \R^n; \vc{a}_i \cdot \vc{p} \leq
  b_i, i = 1, \hdots, k \}$, $k = O(n)$ such that:

  \begin{itemize}
    \item The constraints $\vc{a}_i \cdot \vc{p} \leq \vc{b}$ are either from
      the original box, $p_i \leq M$ or $p_i \geq -M$ or are constraints
      returned by the separation oracle normalized such that $\norm{\vc{a}_i}_2
      = 1$.
    \item The width of $P$ is small, i.e., there is a vector $\vc{a}$ with
      $\norm{\vc{a}}_2 = 1$ such that $$\max_{\vc{p} \in P} \vc{a} \cdot \vc{p} -
      \min_{\vc{p} \in P} \vc{a} \cdot \vc{p} \leq O(n \delta \log(Mn /
      \delta))$$
    \item The algorithm produces a certificate of the previous fact
      in the form of a convex combination
      $t_1, \hdots, t_k$ with $t_i \geq 0$, $\sum_{i=1}^k t_i = 1$ and $k =
      O(1)$ such that:
        \begin{itemize}
          \item $\norm{\sum_{j=1}^k \vc{a}_i}_2 \leq O\left( \frac{\delta}{M}
            \sqrt{n} \log\left( \frac{M}{\delta} \right)\right)$
          \item $\abs{\sum_{j=1}^k b_i} \leq O\left( n \delta
            \log\left( \frac{M}{\delta} \right)\right)$
        \end{itemize}
  \end{itemize}
\end{theorem}

Now, we now show that this result can be used to obtain a convex minimization
algorithm that uses only subgradients:

\begin{theorem}\label{thm:convex_opt}
  Let $f:\R^n \rightarrow \R$ be an $L$-Lipschitz convex function
  equipped with a subgradient oracle with running time $T$. If $f$ has a
  minimizer in $[-M,M]^n$ we can find a point $\bar{\vc{p}}$ such that
  $f(\bar{\vc{p}}) - \min_{\vc{p}} f(\vc{p}) \leq \epsilon$ in time $O(n T
  \log(nML/\epsilon) + n^3 \log^{O(1)}(nML)\log^2 (nML / \epsilon))$.
\end{theorem}

\begin{proof}
  Let $K = \argmin_{\vc{p}} f(\vc{p})$ and use the subgradient as the separation
  oracle, since if $\partial f(\vc{p}_t)$ is the subgradient at $\vc{p}_t$ we
  know that for all $\vc{p} \in K$, it holds that $\partial f(\vc{p}_t) \cdot
  (\vc{p} - \vc{p}_t) \leq 0$. What we will do is to run the algorithm in
  Theorem \ref{thm:lsw} starting from a very large box $[-M', M']$
  for $M'= n^{O(1)} M$ instead of starting from $[-M, M]$, which appears to be more
  natural. The reason we do that is to avoid having the constraints defining
  the bounding box added to $P$.

  Either the algorithm will return a point $\vc{p} \in K$, in which case we are
  done or will return a set $P$ like in statement of the theorem. If we are
  lucky and all constraints added to $P$ are of the type $\partial f(\vc{p}_t)
  \cdot \vc{p} \leq \partial f(\vc{p}_t) \cdot \vc{p}_t$, then we can use the
  certificate $t_1, \hdots, t_k$ to produce a point $\bar{\vc{p}} = \sum_{i=1}^k
  t_i \vc{p}_i$. Now:
  $$\begin{aligned}
    f(\bar{\vc{p}}) - f(\vc{p}^*) & \leq \sum_i t_i f(\vc{p}_i) - f(\vc{p}^*) \leq
  \sum_i \partial f(\vc{p}_i) \cdot (\vc{p}_i - \vc{p}^*) \leq L \cdot \left[
  \norm{\sum_i t_i \vc{a}_i} \cdot \norm{\vc{p}^*} + \abs{\sum_i t_i b_i} \right] \\
    & \leq L \cdot \left[ O\left( \frac{\delta}{M'} \sqrt{n} \log\left(
    \frac{M'}{\delta} \right)\right) \cdot M' \sqrt{n} +  O\left( n \delta
    \log\left( \frac{M'}{\delta} \right)\right)\right]  = \\
    & = O\left( n \delta L \log\left( \frac{M'}{\delta} \right)\right)
  \end{aligned}$$
  Then setting $\delta$ such that $\epsilon = O\left( n \delta L
  \log\left( \frac{M'}{\delta}   \right)\right)$ gives us the desired result.

  To be done, we just need to worry about the case where some of the box
  constraints are present in $P$. In that case, we argue that the weight put by
  the coefficients $t_i$ on the box constraints must be tiny, so it is possible
  to remove those and rescale $t$. Formally, observe that
  $$\sum_i t_i (b_i - \vc{a}_i \cdot \vc{p}^*) \leq \abs{\sum_i t_i b_i} +
  \norm{\sum_i t_i \vc{a}_i}_2 \cdot \norm{\vc{p}^*}_2 = O\left(n \delta
  \log\left( \frac{M'}{\delta} \right) \right) $$
  Now, if $i$ is an index corresponding to a box constraint such as $p_i \leq
  M'$ or $p_i \geq -M'$ then $b_i - \vc{a}_i \cdot \vc{p}^* \geq \abs{M' - M} =
  \Omega(n^{O(1)}M)$. This implies in particular that:
  $$ O\left(n \delta
      \log\left( \frac{M'}{\delta} \right) \right) \geq \sum_i t_i (b_i -
      \vc{a}_i \cdot \vc{p}^*)  \geq  t_i (b_i -
              \vc{a}_i \cdot \vc{p}^*) \geq t_i \Omega(n^{O(1)}M)$$
  so $t_i \leq O\left( \frac{n^{-O(1)} \delta }{M} \log \left( \frac{M'}{\delta}
  \right) \right)$. By choosing a very small $\delta$, say $\delta = O\left(
  \frac{ \epsilon }{ L n^{O(1)} M^{O(1)}}\right)$ we guarantee that is $B$ is the
  set of indices corresponding to box constraints, then $\sum_{i \in B} t_i \leq
  O\left( \frac{ \epsilon }{ n^{O(1)} M^{O(1)}} \right) \leq \frac{1}{2}$ and
  that $\sum_i t_i (b_i - \vc{a}_i \cdot \vc{p}^*) \leq O\left( \frac{ \epsilon
  }{ L n^{O(1)} M^{O(1)}} \right)$. 
  This allows us to define for $i \notin B$, $t'_i = t_i / \left( 1 - \sum_{i \in B}
  t_i \right)$, so we have:
  $$\begin{aligned}
  \frac{\epsilon}{2L} \geq O\left( \frac{ \epsilon }{ L n^{O(1)} M^{O(1)}} \right) & \geq \sum_i t_i (b_i -
  \vc{a}_i \cdot \vc{p}^*) \geq  \sum_{i \notin B} t_i (b_i - \vc{a}_i \cdot
  \vc{p}^*)  \\ & = \left(1-\sum_{i \in B} t_i \right) \left( \sum_{i \notin B} t'_i
  (b_i - \vc{a}_i \cdot \vc{p}^*) \right)  \geq \frac{1}{2} 
  \left( \sum_{i \notin B} t'_i (b_i - \vc{a}_i \cdot \vc{p}^*) \right) 
  \end{aligned}$$
  Now we can repeat the argument in the beginning of the proof with
  $\bar{\vc{p}} = \sum_{i \notin B} t'_i \vc{p}_i$.
\end{proof}

\paragraph{Perturbation, Approximation and Rounding}

The next step is to use the algorithm for finding an approximate solution to
find an exact solution. This is impossible for generic convex
programs, but since the function we are trying to optimize comes from
a linear program, we can do that by exploiting this connection. As done for
linear programs, we will do this in three steps: first we perturb the function
 to be optimized, then we find an approximate solution by Theorem \ref{thm:convex_opt} and finally we
round it to an exact solution in the format of the optimal
solution given by Lemma \ref{lemma:vertices}.

We can perturb the objective of the dual linear program by changing it to:
\begin{equation}\label{perturbed_dual_program}\tag{PD}
\min u + \vc{p} \cdot (\vc{s} + \vc{r}) \text{ s.t. } u \geq \sum_i
v_i(x^{(i)}) - \vc{p} \cdot x^{(i)}, \forall x^{(i)} \in \dbracket{s}
\end{equation}
We want to specify the vector $\vc{r}$ in such a way that the optimal solution
to this linear program is still the optimal solution to the original program,
but also in a way that the optimal solution becomes unique. First we observe the
following:

\begin{lemma}\label{lemma:small_perturbation}
If $r_i < (nS)^{-(2n+1)}$ then an optimal solution of the perturbed
program is also an optimal solution to the original program.
\end{lemma}

\begin{proof} 
Let $\C$ be the set of vertices (basic feasible solutions) of the dual LP.
We know by Lemma \ref{lemma:vertices} we know that the coordinates of
those vertices must be of the form $a_i / b_i$ for integers $a_i, b_i$ such
that $0 \leq b_i \leq (nS)^n$. For any linear objective function, the optimal solution
must be a point in $\C$. Since the original objective has integral coefficients,
when evaluated on any vertex, the objective is a fraction with denominator $b_i
\leq (nS)^n$. Therefore, the difference between the objective evaluated at an
optimal vertex and the objective evaluated at a suboptimal one is at
least $\abs{\frac{a_i}{b_i}-\frac{a'_i}{b'i}} = \abs{\frac{a_i b'_i -
a'_i b_i}{b_i b'_i}} \geq \frac{1}{(nS)^{2n}}$. Therefore, if $r_i <
(nS)^{-(2n+1)}$, then its effect in the objective function is at most $nS \cdot
(nS)^{-(2n+1)} \leq (nS)^{-2n}$, so it can't cause a suboptimal solution to
the original program to become an optimal solution of the perturbed program.
\end{proof}

Our final ingredient is a lemma by Klivans and Spielman
\cite{klivans2001randomness} which is in the spirit
of the Isolation Lemma of Mulmuley, Vazirani and Vazirani \cite{IsolationLemma}.

\begin{lemma}[Klivans and Spielman (Lemma 4 in
\cite{klivans2001randomness})]\label{lemma:ks}
Let $\C$ be a set of points in $\R^n$ where all coordinates assume at
most $K$ distinct values. Then if $\vc{r}$ is a vector with coordinates sampled
at random from $\{0, 1, \hdots, Kn/\epsilon\}$, then with probability
$1-\epsilon$, there is a unique $\vc{p} \in \C$ minimizing $\vc{r} \cdot
\vc{p}$.
\end{lemma}

We note that although Lemma 4 in \cite{klivans2001randomness} is stated as $\C$
having coordinates in $\{0, 1, \hdots, K-1\}$, the proof only uses the fact that
the coordinates of $\C$ assume at most $K$ distinct values.

\begin{lemma}\label{lemma:good_perturb}
If $r_j = z_j / \left[ M n (nS)^{2n+1}\right]$ where $z_j$ is drawn uniformly
from $\{0, \hdots, n M(nS)^{2n} - 1\}$, then with  probability $\frac{1}{2}$, the dual
program has an unique minimizer and this minimizer is a minimizer of the
original program.
\end{lemma}

\begin{proof}
Let $\C$ be the set of vertices (basic feasible solutions) of the dual linear
program in $[-M, M]^n$. All the minimizers of the original dual program are
guaranteed to be in this set because of Lemma \ref{lemma:initial_box}. Lemmas
\ref{lemma:initial_box} and \ref{lemma:vertices} tell us that the coordinates of
points in $\C$ are of the form $a_i / b_i$ where $a_i$ and $b_i$ are
integers such that $1 \leq b_i \leq (nS)^n$ and $\abs{a_i} \leq M b_i$.
So there are at most $\sum_{b=1}^{(nS)^n} b \cdot 2M \leq M (nS)^n$
coordinates.

Now, Lemma \ref{lemma:small_perturbation} guarantees that the magnitude of the
perturbation will prevent suboptimal vertices in the original program to become
optimal vertices in the perturbed program. Lemma \ref{lemma:ks} guarantees that
with half probability the vertex minimizing the objective is unique.\end{proof}

The perturbed dual program translates to the following perturbed objective
function:

\begin{equation}\label{perturbed_convex_f}\tag{PC}
\hat{f}(\vc{p}) =  \sum_{i=1}^{m}\left(\max_{x\in\dbracket{\vc{s}}}v_{i}(x)-\vc{p}\cdot
  x\right)+\vc{p} \cdot (\vc{r} + \vc{s}).
\end{equation}

Since the subgradient of $\hat{f}$ can be computed from the aggregate demand
oracle $\partial \hat{f}(\vc{p}) = \vc{s} + \vc{r} - \vc{d}(\vc{p})$ we can use
Theorem \ref{thm:convex_opt} to find an $\epsilon$-minimizer.

\begin{lemma}\label{lemma:cloness_guarantee}
For $\epsilon = (nMS)^{-O(n)}$, if $\vc{p}$ is an $\epsilon$-approximation to the
optimal value of $\hat{f}$, i.e., $\hat{f}(\vc{p}) - \hat{f}(\vc{p}^*) \leq
  \epsilon$ then the optimal solution is the only point $\vc{p^*}$ such that $\norm{\vc{p} -
  \vc{p^*}}_2 \leq \frac{1}{(nMS)^{O(n)}}$ and has coordinates of the form
  described in Lemma \ref{lemma:vertices}.
\end{lemma}

\begin{proof}
  Let $\C$ be the set of vertices of the dual linear program, which are points of
  the form $(u_i, \vc{p}_i)$ and let $\vc{p}$ be an $\epsilon$-approximation to
  $\hat{f}$. Now, if $u = \sum_i \left(\max_{x} v_i(x)-\vc{p}\cdot x\right)$
  then $(u, \vc{p})$ is feasible in the dual linear program, and in particular,
  it can be written as a convex combination of points in $\C$, i.e., $(u,\vc{p}) =
  \sum_i t_i (u_i, \vc{p}_i)$ for $t_i \geq 0$ and $\sum_i t_i = 1$. There is
  only vector in $\C$, call it $(u^*, \vc{p}^*)$ for which the objective
  evaluate to $\hat{f}^*$. For all other vertices, the objective evaluate to at
  least $\hat{f}^* + \frac{1}{M n (nS)^{2n+1}} \cdot \frac{1}{(nS)^n} =
  \hat{f}^* + \frac{1}{M n (nS)^{3n+1}}$ due to Lemmas \ref{lemma:vertices} and
  \ref{lemma:good_perturb}. Therefore, by evaluating  $(u,\vc{p}) =
  \sum_i t_i (u_i, \vc{p}_i)$ on the linear objective of the perturbed dual
  program, we get:
  $$\hat{f}^* + \epsilon \geq \hat{f}(\vc{p}) \geq t_* \hat{f}^* + (1-t_*)
  \left( \hat{f}^* + \frac{1}{M n (nS)^{3n+1}} \right)$$
  where $t_*$ is the weight put on $(u^*, \vc{p}^*)$ by the convex combination,
  therefore, if $\epsilon = (nMS)^{-O(n)}$ then $t^* \geq 1 - (nMS)^{-O(n)}$. In
  particular:
  $$\norm{\vc{p} - \vc{p}^*}_2 \leq \norm{\sum_{i \neq *} t_i (\vc{p}_i -
  \vc{p}^*)}_2 \leq (1-t_*) \cdot \max_{i \neq *} \norm{\vc{p}_i - \vc{p}^*}
  \leq (nMS)^{-O(n)}$$
  Therefore, $\vc{p}^*$ is in a ball of radius $(nMS)^{-O(n)}$ around $\vc{p}$.
  Also, for any other point $\vc{p}' \neq \vc{p}^*$ with coordinates of the form
  $a_i / b_i$ with $b_i \leq (nS)^n$ can be in this ball, since for two distinct
  points of this type differ in at least one coordinate by at least $(nS)^{-n}$
  so their distance is at least that much.
\end{proof}

Putting it all together:

\begin{theorem}\label{thm:algo_we_general}
  There is an algorithm of running time $O(n^2 T_{AD} \log(SMn) + n^5
  \log^{O(1)}(SMn))$ to compute an exact vector of Walrasian prices whenever it
  exists using only access to an aggregate demand oracle.
\end{theorem}

\begin{proof}
  From the potential function $f$ of the market, use Lemma
  \ref{lemma:good_perturb} to construct a perturbed potential function
  $\hat{f}$. Then use Theorem \ref{thm:convex_opt} to optimize $\hat{f}$
  with $\epsilon = (nMS)^{-O(n)}$. Finally, use the guarantee to argue that the
  exact optimal value is the only point in the format of Lemma
  \ref{lemma:vertices} in the ball of radius $(nMS)^{-O(n)}$ around the
  $\epsilon$-approximate minimizer. At this point, we can round the solution to
  the exact point by using the method of continuous fractions in Kozlov et al
  \cite{kozlov1980polynomial} (see
  \cite{schrijver1998theory} for a complete exposition of this and related
  methods of rounding).
\end{proof}

\section{Walrasian Equilibrium for Gross Substitutes in $\tilde{O}(n T_{AD} +
n^3)$}

In the previous section we discussed how Walrasian equilibria can be computed
without any assumptions on the valuation function using only an aggregate demand
oracle. The focus was to address various difficulties in
applying optimization tools for this problem.

Here we remain in the aggregate demand oracle model and
focus on computing Walrasian equilbria for markets typically studied in
economics, which are those where buyers have \emph{gross substitute
valuations}. The development of the theory of gross substitute valuations is
intertwined with the development of the theory of Walrasian equilibrium for
markets with discrete goods. In particular, it is the largest class of valuation
functions that is closed under perturbation by an additive function\footnote{A class
$\mathcal{C}$ of valuation functions is closed under perturbations by additive
functions if for every $v \in \mathcal{C}$ and every vector $\vc{w} \in \R^n$, the
function $w(x) = v(x) + \vc{w} \cdot x$ is also in $\mathcal{C}$.} for which
Walrasian equilibria always exist.

\comment{
In this
section, we restrict our attention to the class of gross substitute valuations,
for which  a Walrasian equilibrium is known to always exist, as originally
showed by Kelso and Crawford \cite{KelsoCrawford}. Gul and Stachetti
\cite{GulStachetti} later show that gross substitutes are the only class of
functions closed under perturbation by an additive function\footnote{A class
$\mathcal{C}$ of valuation functions is closed under perturbations by additive
functions if for every $v \in \mathcal{C}$ and every vector $\vc{w} \in \R^n$, the
function $w(x) = v(x) + \vc{w} \cdot x$ is also in $\mathcal{C}$.} for which
Walrasian equilibrium exists. Recently, other classes of functions were shown to
have Walrasian equilibria \cite{SunYang2006, ozan14, ozan15, BenZwi}.
}

Gross substitutes play a central role in economics. Hatfield and
Milgrom \cite{HatfieldMilgrom} argue that most important examples of valuation
functions arising in matching markets belong to the class of gross substitutes.
Gross substitutes have been used to guarantee the convergence of salary
adjustment processes in the job market \cite{KelsoCrawford},
to guarantee the existence of stable matchings
in a variety of settings \cite{Roth84, Kojima}, to show the stability of trading
networks \cite{HatfieldKNOW15}, to design combinatorial auctions
\cite{AusubelMilgrom,Shioura13} and even to analyze settings with
complementarities \cite{SunYang2006,HatfieldK15}.
  
The concept of gross substitutes has been re-discovered in different areas from
different perspectives: Dress and Wenzel \cite{DressWenzel92} propose the
concept of \emph{valuated matroids} as a generalization to the
Grassmann-Pl\"{u}cker relations in $p$-adic analysis. Dress and Terhalle
\cite{DressTerhalle_WellLayered} define the concept of \emph{matroidal maps}
which are the exact class of functions that can be optimized by greedy
algorithms. Murota \cite{Murota96} generalized the concept of convex functions
to discrete lattices, which gave birth to the theory known as Discrete Convex
Analysis. One of the central objects in the Discrete Convex Analysis are
$M$-concave and $M^\natural$-concave functions (the latter class was introduced by
Murota and Shioura \cite{MurotaShioura}).

Surprisingly, gross substitutes, valuated matroids, matroidal maps and
$M^\natural$-concave functions are all equivalent definitions -- despite being
originated from very different motivations. We refer to \cite{gs_survey} for a
detailed historic description and a survey on their equivalence.\\

Before presenting our algorithms, we first give a quick summary of
the standard facts about gross substitutes that are needed. For a more comprehensive
introduction, please see \cite{gs_survey}.

\subsection{A crash course on gross substitutes}

First we define gross substitute valuations on the hypercube $\{0,1\}^{[n]}$ and then
we extend the defintion to $\dbracket{\vc{s}}$ for any supply vector $\vc{s}$.
When talking about functions defined on the hypercube, we will often identify
vectors $x \in \{0,1\}^{[n]}$ with the set $S = \{i \in [n]: x_i = 1\}$, so 
we write $v(S)$ where $S \subseteq [n]$ meaning $v(\vc{1}_S)$ where $\vc{1}_S$
is the indicator vector of $S$.

\comment{First of all, we need a change of notations. Gross substitutes are 
functions defined on a ground \emph{set} rather than 
\emph{multiset}. Thus in this
section we have supply $s_{j}=1$ for all item $j$ and our valuation
function $v_{i}:2^{[n]}\rightarrow \Z$ has ground set
$[n]$. Each buyer $i$ would choose a subset $S$ which solves 
$$
\max_{S\subseteq[n]}v_{i}(S)-p(S)
$$
where $p(S)=\sum_{j\in S}p_{j}$. So the demand set $D(i,p)$ consists
of subsets $S\subseteq[n]$ rather than vectors $x\in\dbracket{\mathbf{s}}$.

Economically, the multiset setting can be reduced to the single-unit
supply setting by treating each unit of items as a distinct entity.
There has been some recent effort dedicated to extending gross substitutes to the
multi-unit setting \cite{shioura2015gross}. We will draw on the tools 
developed there in the next section to give similar results on market
equilibrium computation in the multi-unit gross substitute setting.}

Next we define three classes of valuation functions. The reader which saw the
spoilers in the previous subsection will already suspect that they are
equivalent.

\begin{definition}(Gross substitutes, Kelso and Crawford \cite{KelsoCrawford})
\label{def:gs}A function $v:2^{[n]}\rightarrow\Z$ is
a gross substitute (GS) if for any price $\vc{p}\in\R^{n}$ and
$S\in D(v,\vc{p})$, any price $\vc{p}'\geq \vc{p}$ (entry-wise) has some $S'\in
D(v,\vc{p}')$
satisfying
$
S\cap\{j:p_{j}=p_{j}'\}\subseteq S'.
$
\end{definition}

In other words, price increases for some items can't affect the
purchasing decisions for the items whose price stayed the same.

The second definition concerns when the demand oracle problem $\max
v(S)-p(S)$ can be solved efficiently:

\begin{definition}
\label{def:matroidal}[Matroidal Maps, Dress and Terhalle
\cite{DressTerhalle_WellLayered}]A function $v:2^{[n]}\rightarrow\Z$
is called matroidal if for any price $\vc{p}\in\R^{n}$, the set
$S$ obtained by the following greedy algorithm solves
$\max_{S\subseteq[n]}v(S)-p(S)$:

\textbf{Greedy algorithm:} Start with the empty set $S = \emptyset$. While
$\abs{S} < n$ and there exists $i \notin S$ such that $v(S \cup i) -p_i - v(S) > 0$,
then update $S \leftarrow S \cup i^*$ where $i^* = \argmax_{i \in [n] \setminus
  S} v(S \cup i) -p_i - v(S)$ (break ties arbitrarily).
\end{definition}

The third definition generalizes the concept of convexity and concavity to the
hypercube. We can define convexity for continuous functions $f:\R^n \rightarrow
\R$ as being function such that for all vectors $\vc{v} \in \R^n$, if $x^*$ is
a local minimum of $f(x) - \vc{v} \cdot x$, then $x^*$ is also a global
minimum. This definition generalizes naturally to the hypercube as follows:

\begin{definition}
\label{def:discrete_convex}[Discrete Concavity, Murota \cite{Murota96}]
A function $v:2^{[n]}\rightarrow\Z$ is discrete concave function if local
optimality implies global optimality for all price vectors $\vc{p}$.
Formally: if $S \subseteq [n]$ is such
that for all $i \in S$ and $j \notin S$, $v(S) \geq v(S \cup j) - p_j$, $v(S) \geq
v(S \setminus i) + p_i$ and $v(S) \geq v(S \cup j \setminus i) + p_i - p_j$,
then $v(S) - \vc{p}(S) \geq v(T) - \vc{p}(T), \forall T \subseteq [n]$.
\end{definition}

Discrete concave functions have two important properties:
(i) the demand oracle has a succinct certificate of optimality and
(ii) the function can be optimized via local search for any given prices.

Those definitions arose independently in different communities and,
amazingly enough, those definitions turned out to be equivalent:

\begin{theorem}[Fujishige and Yang \cite{FujishigeYang}]
  A valuation function is in gross substitutes iff it is a
  matroidal map and iff it is a discrete concave valuation.
\end{theorem}

\comment{
  We note that there are a few more starkly different definitions which
are also equivalent to GS but are not needed for our purpose. Indeed,
these unexpected equivalences demonstrate that GS is an unusually
richly structured class of discrete functions that should perhaps
deserve more attention.}

The concept of gross substitutes generalizes naturally to multi-unit valuations:
given any function $v:\dbracket{\vc{s}} \rightarrow \Z$ we can translate it to a
single-unit valuation function $\tilde{v}:2^{[\sum_i s_i]} \rightarrow \Z$ by
treating each of the $s_i$ copies of item $i$ as a single item. We say that a
multi-unit valuation function $v$ is gross substitutes if its corresponding
single item version $\tilde{v}$ is in gross substitutes. We refer to the
excellent survey by Shioura and Tamura \cite{shioura2015gross} on gross
substitutability for multi-unit valuations.

An important property of gross substitutes is:

\begin{theorem}[Kelso and Crawford \cite{KelsoCrawford}]\label{thm:kelso_crawford}
  If all buyers have
  gross substitute valuations, then a Walrasian equilibrium always exists.
\end{theorem}

\subsection{Walrasian Prices form an integral polytope}

By Theorem \ref{thm:kelso_crawford}, the set of Walrasian prices
is non-empty. We also know that is forms a convex set polyhedral set, since
they are the set of minimzers of a convex function that comes from a linear
program. Perhaps a more direct way to see it is that given an optimal
allocation $\vc{x}$, the set of Walrasian prices can be characterized by the
following set of inequalities:
$$P = \{ \vc{p} \in \R^n \text{ } \vert \text{ } v_i(x^{(i)}) - \vc{p} \cdot x^{(i)}
\geq v_i(x) - \vc{p} \cdot x, \forall i \in [m], x \in \dbracket{s} \}$$
Lemma \ref{lemma:vertices} provided a good characterization of the vertices of
this polytope in the general case. For gross substitutes, however, there is an
even nicer characterization:

\begin{lemma}\label{lemma:vertices_gs}
  If buyer valuations are gross substitutes, then
  all the vertices of the feasible region of the dual program \ref{dual_program}
  (defined in Section \ref{sec:general}) have integral coordinates. In
  particular, the set of Walrasian prices form an integral polytope.
\end{lemma}

\begin{proof}
  Let $(u, \vc{p})$ be a non-integral feasible point of the dual program
  \ref{dual_program}. We will write it as a convex combination of integral
  points. Let $x^{(i)} \in \argmax v_i(x) - \vc{p} \cdot x$ and $w = u - \sum_i
  v_i(x^{(i)}) - \vc{p} \cdot x_i$.

  Now, we define a distribution over integral points in the following way:
  sample a random threshold $\theta \in [0,1]$. If $\theta < p_i - \floor{p_i}$,
  set $\hat{p}_i = \ceil{p_i}$, otherwise set $\hat{p}_i = \floor{p_i}$.
  Similarly, if $\theta < w - \floor{w}$, set $\hat{w} = \ceil{w}$ otherwise set
  $\hat{w} = \floor{w}$. Set $\hat{u} = \hat{w} + \sum_i v_i(x^{(i)}) -
  \hat{\vc{p}} \cdot x^{(i)}$. It is easy to see that $(\hat{u}, \hat{\vc{p}})$
  are integral and that $\E[ (\hat{u}, \hat{\vc{p}}) ] = (u, \vc{p})$. We are
  only left to prove that $(\hat{u}, \hat{\vc{p}})$ are feasible.

  We know that $\hat{u} \geq  \sum_i v_i(x^{(i)}) - \hat{\vc{p}} \cdot
  x^{(i)}$ since $\hat{w} \geq 0$. If we show that $x^{(i)} \in \arg \max_x
  v_i(x) - \vc{p} \cdot x$, then we are done, since it automatically implies
  that all other constraints in the dual program \ref{dual_program} are
  satisfied. To see this notice that since $x^{(i)} \in \argmax v_i(x) - \vc{p}
  \cdot x$, then for all items $j$ and $k$ such that $x^{(i)}_j < s_j$ and
  $x^{(i)}_k > 0$ we have that:
  $$v_i(x^{(i)}) \geq v_i(x^{(i)} + \vc{1}_j) - p_j, \text{ } v_i(x^{(i)}) \geq
  v_i(x^{(i)}  - \vc{1}_k) + p_k \text{ and }  v_i(x^{(i)}) \geq v_i(x^{(i)} 
  - \vc{1}_k  +  \vc{1}_j) + p_k - p_j$$
  Since the valuations are integer valued, it is simple to check that rounding
  using a threshold won't violate any of the above inequalities. Thus:
  $$v_i(x^{(i)}) \geq v_i(x^{(i)} + \vc{1}_j) - \hat{p}_j, \text{ } v_i(x^{(i)}) \geq
  v_i(x^{(i)}  - \vc{1}_k) + \hat{p}_k \text{ and }  v_i(x^{(i)}) \geq v_i(x^{(i)} 
  - \vc{1}_k  +  \vc{1}_j) + \hat{p}_k - \hat{p}_j$$
  Therefore under price vector $\hat{\vc{p}}$, a buyer can't improve his utility
  from $x^{(i)}$ by adding, removing or swapping an item. Since gross
  substitutes are equivalent to discrete concavity (Definition
  \ref{def:discrete_convex}), local optimality implies global optimality, i.e.,
  $x^{(i)} \in \arg \max_x v_i(x) - \hat{\vc{p}} \cdot x$.
\end{proof}

Another important property proved by Gul and Stachetti \cite{GulStachetti} is
that the set of Walrasian prices forms a lattice.

\begin{theorem}[Gul and Stachetti \cite{GulStachetti}]\label{thm:lattice}
  If buyer valuations are gross substitutes, then the set of Walrasian prices
  form a lattice, i.e., if $\vc{p}$ and $\vc{p}'$ then $\bar{\vc{p}}$ and
  $\underline{\vc{p}}$ are also Walrasian prices for $\bar{p}_i = \max(p_i,
  p'_i)$ and $\underline{p}_i = \min(p_i, p'_i)$.
\end{theorem}

\subsection{A simpler and faster algorithm for gross substitutes}

Using the fact that the set of Walrasian prices is an integral polytope with a
lattice structure we simplify the algorithm described in the previous section
and improve its running time.
First, since we have a lattice structure we no longer need to randomly perturb
the objective function to make the solution unique. A simple and
deterministic perturbation suffices. Integrality also allows us to round to an
optimal solution from an approximate solution of smaller accuracy (i.e. larger
$\epsilon$).

\begin{lemma}\label{lemma:better_perturbation}
  If valuations are gross substitutes, then by taking $r_j = \frac{1}{2Sn}$
  in the perturbed dual program \ref{perturbed_dual_program}, its optimal solution
  is unique and also optimal for the original dual program
  \ref{dual_program}.
\end{lemma}

\begin{proof}
  Since all the vertices of the polytope are integral and the coefficients are
  at most $S$, a perturbation of $r_j = \frac{1}{2Sn}$ can't affect the
  objective by more than half. So it can't cause a suboptimal vertex to become
  optimal. Also, since the set of Walrasian prices form a lattice, there is a
  Walrasian price $\bar{\vc{p}}$ such that $\bar{\vc{p}} \geq \vc{p}$ for every
  Walrasian price $\vc{p}$. Therefore this must be the unique optimal solution to the
  perturbed program.
\end{proof}

The previous lemma allows us to prove a better version of Lemma
\ref{lemma:cloness_guarantee}:

\begin{lemma}\label{lemma:cloness_guarantee_gs}
For $\epsilon < 1/(4nMS)$, if $\vc{p}$ is an $\epsilon$-approximation to the
optimal value of $\hat{f}$, i.e., $\hat{f}(\vc{p}) - \hat{f}(\vc{p}^*) \leq
  \epsilon$ then the optimal solution is the only integral point $\vc{p^*}$ such that $\norm{\vc{p} -
  \vc{p^*}}_2 < \frac{1}{2}$.
\end{lemma}

\begin{proof} It follows exact the same proof of Lemma
  \ref{lemma:cloness_guarantee} when the better guarantees provided by Lemma
  \ref{lemma:better_perturbation} are used.
\end{proof}

Putting it all together we have:

\begin{theorem}\label{thm:algo_ad_gs}
  There is an algorithm of running time $O(n T_{AD} \log(SMn) + n^3
  \log^{O(1)}(SMn))$ to compute an exact vector of Walrasian prices in a market
  where buyers have gross substitute valuations.
\end{theorem}

\begin{proof}
  Same proof as in Theorem \ref{thm:algo_we_general} with $\epsilon = 1/(5nMS)$.
  Also, since the optimal price vector is integral, instead of using the method
  of continuous fractions to round a solution, it is enough to round each
  component to the nearest integer.
\end{proof}

\section{Walrasian Equilibrium for Gross Substitutes in $\tilde{O}((mn + n^3)\cdot
T_V)$}\label{sec:gs_value_oracle}

We now move from the macroscopic view of the market to a microscopic view. We
assume access to the market via a \emph{value oracle}, i.e, given a
certain buyer $i$ and a bundle $S \subseteq [n]$ of goods,
we can query the value of $v_i(S)$. We also assume from this point on that
the supply of each good is one, or in other words, that the valuation
functions are defined on the hypercube.

The fact that the demand of each buyer for any given price can be computed by
the greedy algorithm (Definition \ref{def:matroidal}) lets us simulate
the aggregate demand oracle by the value oracle model.

\begin{lemma} \label{lem:gsdemand}
  The outcome of the aggregate demand oracle can be computed from in time
  $O(mn^{2}T_{V})$, where $T_V$ is the running time of the value oracle.
\end{lemma}
\begin{proof}
  The number of queries required for the greedy algorithm described in
  Definition \ref{def:matroidal} to compute $S_i \in \arg \max_S v(S) -
  \vc{p}(S)$ is $n \cdot (\abs{S_i} + 1) T_V \leq O(n^2 T_V)$. Since there are
  $m$ buyers, the total time to compute the demand of all buyers is $O(mn^2
  T_V)$. The aggregate demand oracle simply outputs $d(p)$ where
  $d_{j}(p)=\#\{i:j\in S^*_i\}$.
\end{proof}

Now we can plug Lemma \ref{lem:gsdemand} directly into Theorem
\ref{thm:algo_ad_gs} and obtain a running time of $\tilde{O}(mn^3 T_{V})$. In the rest of this section, we show how to improve this to $\tilde{O}((mn + n^3) T_{V})$.

\subsection{Faster Algorithm via regularization}

The idea to improve the running time from $\tilde{O}(mn^3 T_{V})$ to
$\tilde{O}((mn + n^3) T_{V})$ is to regularize the objective function.
As with the use of regularizers in other context in optimization,
this is to penalize the algorithm for being too aggressive. The bound
of $O(mn^{2})$ value oracle calls per iteration of the cutting plane
algorithm is so costly precisely because we are trying to take
an aggressively large step.

To provide some intuition, imagine that we have a price $\vc{p}$ that is very
close to optimal and that $S_i$ are the set of items demanded by the buyers at
price $\vc{p}$. Intuitively, if $\vc{p}$ is close to a Walrasian price then the
sets $S_1, \hdots, S_m$ should be almost disjoint, which means that the total
cost of the greedy algorithm should be $\sum_i n (\abs{S_i}+1) \approx mn +
n^2$. So when the prices are good, oracle calls should be cheaper. This tells us
that when prices are good, fewer calls to the value oracle suffice to compute
the aggregate demand oracle. When prices are far from equilibrium, perhaps a more
crude approximation to the aggregate demand oracle is enough.

Based on this idea we define the following regularized objective function:

  \begin{equation}\label{reg_convex_f}\tag{RC}
    \tilde{f}(\vc{p})=\max_{\sum_i \abs{S_i} =
    n}\left[\sum_{i=1}^{m}v_{i}(S_i)-\vc{p}(S_i)\right] +\vc{p}([n]).
  \end{equation}

The regularization consists of taking the maximum over all sets $(S_1, \hdots,
S_n)$ such that $\sum_i \abs{S_i} = n$. Without this restriction, we have the
original market potential function $f$. The new function $\tilde{f}$ has three
important properties: (i) it is still convex, since it is a maximum over linear
functions in $\vc{p}$ and therefore we can minimize it easily; (ii) its set of
minimizers is the same as the set of minimizers of $f$  and (iii) subgradients
are cheaper to compute. Intuitively, $\tilde{f}$ is very close to $f$ when
$\vc{p}$ is close to
equilibrium prices but only a crude estimate when $\vc{p}$ is far from
equilibrium. Next we show those statements formally and present an algorithm for
computing the subgradient of $f$.

We give an alternate form of $\tilde{f}$ which is nicer to work with
algorithmically. The next lemma shows that
$$
\tilde{f}(\vc{p})=\sum_{i=1}^{m}\left(\max_{S\subseteq[n]}v_{i}(S)-
\vc{p}(S)-\gamma\cdot|S|\right)+\vc{p}([n])+\gamma\cdot n
$$
for some $\gamma$ that depends on $\vc{p}$ (and the tiebreaking rule
used by the greedy algorithm). Among other things, this formulation
of $\tilde{f}$ resembles common regularizers used in optimization
better. One can think of it as if $\vc{p}$ is changed to $\vc{p}+\gamma\cdot\mathbf{1}_{[n]}$.
\begin{lemma}
\label{lem:regularize_lemma}Suppose $v_{i}(j)$ are given and stored
as $n$ lists sorted in decreasing order $\texttt{L}_j = \{v_{i}(j)\}_{i=1..m}$.
With a running time of $n^{2}\cdot T_{V}+\tilde{O}(n^{2})$
\footnote{Assuming the cost of initializing $S_{i}^{*}=\emptyset$ is not needed.
This is acceptable here because our algorithm would only use $S_{i}^{*}$
to compute the subgradient which $S_{i}^{*}=\emptyset$ has no effect
on.
}, given price $\vc{p}$, there is an algorithm, which we call
\textsc{AllGreedy}, that finds \begin{enumerate}
\item $S_{1}^{*},\ldots,S_{m}^{*}$ maximizing
$
\max_{\sum_{i}|S_{i}|=n}\left(\sum_{i=1}^{m}v_{i}(S_{i})-\vc{p}(S_{i})\right).
$

\item $\gamma$ such that for all $i$, $S_{i}^{*}\in D(i,\vc{p}+\gamma\cdot\vc{1}_{[n]})$.
Moreover, for any $\gamma'>\gamma$ and $S_{i}'\in D(i,\vc{p}+\gamma'\cdot\vc{1}_{[n]})$,
we have $\sum_{i}|S_{i}'|<n$.
\item $\tilde{f}(\vc{p})=f(\vc{p}+\gamma\cdot\vc{1}_{[n]})$.
\end{enumerate}
\end{lemma}
\begin{proof}
First, we define the algorithm \textsc{AllGreedy}.
The algorithm starts with a
very large value of $\gamma$ such that
$D(i,\vc{p}+\gamma\cdot\vc{1}_{[n]})=\{\emptyset\}$ for all agents $i$. Then we
gradually decrease $\gamma$ keeping at each step a set $S^*_i(\gamma) \in
D(i,\vc{p}+\gamma\cdot\vc{1}_{[n]})$ that monotonically grow as $\gamma$
decreases, in the sense that for $\gamma_1 > \gamma_2$, $S_i^*(\gamma_1)
\subseteq \S_i^*(\gamma_2)$. 
We stop the algorithm as $\sum_i \abs{S^*_i(\gamma)}$ reaches $n$.

The algorithm is best visualized as a continuous process, although it admits a
very efficient discrete implementation as we will see in a moment. Before, we
note that we can use the Greedy algorithm to compute $S^*_i(\gamma) \in
D(i,\vc{p}+\gamma\cdot\vc{1}_{[n]})$ and if we fix the same tie breaking, the
order in which we add the elements is the same for every $\gamma$, the only
thing that changes is the stopping criteria (the larger $\gamma$ is, the later
we stop).

So this procedure can be implemented by running a greedy algorithm in parallel
for each agent $i$. Initially $\gamma$ is very large and all $S_i^*(\gamma) =
\emptyset$. Then in any given moment, we can compute what is the largest value
of $\gamma$ for which it is possible to add one more item to the demanded set of
$i$. This is simply the largest marginal of an $i$ for the next item:
$$\max_i \max_{j \notin S^*_i} v_i(S^*_i \cup j) - v_i(S^*_i) - p_j$$
We can decrease $\gamma$ to this value and advance one of the agent's greedy
algorithm one step further.

We need to argue that it satisfies the three properties in the lemma and that it
can be implemented in $n^2 T_V + O(n^2)$ time. For the running time, the
algorithm can by updating lists $\texttt{L}_j$ such that in each iteration, it
is a sorted list of $v_i(S^*_i \cup j) - v_i(S_i^*)$. Since all sets start as
the empty set, this is correct in the beginning of the algorithm. Now, in each
iteration, we can scan all the lists to find the next largest marginal, taking
$O(n)$ to inspect the top of each list. This gives us the next value of $\gamma$
and the pair $i,j$ to update $S^*_i \leftarrow S^*_i \cup j$. Now, after the
update, we go through each list $\texttt{L}_k$ updating the value of the
marginal of agent $i$ for $k$, since $S_i^*$ was updated. This takes
$O(\log(m))$ for each list, so in total, this process takes $n T_V + \tilde{O}(n)$.
Since there are at most $n$ iterations, the overall running time is $n^2 T_V +
\tilde{O}(n^2)$.

Now, for three properties in the lemma, property 2 is true by construction. For
properties 1 and 3, consider the following chain of inequalities:

$$\begin{aligned}
\tilde{f}(\vc{p}) & =
\max_{\sum_{i}|S_{i}|=n}\left(\sum_{i=1}^{m}v_{i}(S_{i})-\vc{p}(S_{i})\right)
+ \vc{p}([n]) \\
   & =
  \max_{\sum_{i}|S_{i}|=n}\left(\sum_{i=1}^{m}v_{i}(S_{i})-\vc{p}(S_{i})-\gamma\cdot|S_{i}|\right)
  + \vc{p}([n]) +\gamma\cdot n\\
 & \leq
 \max_{S_{i}\subseteq[n]}\left(\sum_{i=1}^{m}(v_{i}(S_{i})-\vc{p}(S_{i})-\gamma\cdot|S_{i}|\right)+\vc{p}([n])
 +\gamma\cdot n
  = f(\vc{p} + \gamma \cdot \vc{1}_{[n]})\\
 & =
 \sum_{i=1}^{m}\left[v_{i}(S_{i}^{*})-\vc{p}(S_{i}^{*})-\gamma\cdot|S_{i}^{*}|\right]+
 \vc{p}([n]) +\gamma\cdot n
  =  \sum_{i=1}^{m} \left[ v_{i}(S_{i}^{*})-\vc{p}(S_{i}^{*}) \right] +
  \vc{p}([n])\\
 & \leq
 \max_{\sum_{i}|S_{i}|=n}\left(\sum_{i=1}^{m}v_{i}(S_{i})-\vc{p}(S_{i})\right) +
 \vc{p}([n]) = \tilde{f}(\vc{p})
\end{aligned}$$

Hence, all inequalities hold with equality, which means in particular that
$\tilde{f}(\vc{p}) = f(\vc{p} + \gamma \cdot \vc{1}_{[n]})$ and $S^*_1, \hdots,
S_m^*$ maximize
$\max_{\sum_{i}|S_{i}|=n}\left(\sum_{i=1}^{m}v_{i}(S_{i})-\vc{p}(S_{i})\right)$.
\end{proof}
\begin{corollary}
\label{cor:reg_sep}Suppose $v_{i}(j)$ are given and stored as $n$
sorted lists $\{v_{i}(j)\}_{j}$ each of which has $m$ elements.
Then the greedy algorithm computes a subgradient of $\tilde{f}$ in
$n^{2}\cdot T_{V}+\tilde{O}(n^{2})$ time.\end{corollary}
\begin{proof}
This follows directly from Lemma \ref{lem:regularize_lemma} as the gradient of
$\sum_{i=1}^{m}\left(v_{i}(S_{i}^{*})-\vc{p}(S_{i}^{*})\right) - \vc{p}([n])$ is
a subgradient of $\tilde{f}$.\end{proof}
\begin{corollary}
\label{cor:reg_minimizer}If $\vc{p}^{*}$ minimizes $\tilde{f}$, then
$\vc{p}^{*}+\gamma\cdot\vc{1}_{[n]}$ is an equilibrium price. Here
$\gamma$ is defined as in Lemma \ref{lem:regularize_lemma} with
respect to $\vc{p}^{*}$. Conversely, any Walrasian price $\vc{p}^{\eql}$ is a minimizer of $\tilde{f}$.
\end{corollary}

The proof of the previous corollary is given in the appendix.

\begin{theorem}
For gross substitutes, we can find an equilibrium price in
time $mn\cdot T_{V}+O(mn\log m+n^{3}\log(mnM)\cdot T_{V}+n^{3}\log^{O(1)}(mnM))$.
\end{theorem}

\begin{proof}
  Corollary \ref{cor:reg_minimizer} says that it is enough to find a minimizer
  of $\tilde{f}$. The exact same method used in Theorem \ref{thm:algo_ad_gs} can
  be used to solve $\tilde{f}$ approximately and then round it to an optimal
  solution.

  To bound the overall running time, we note that:
  computing $v_{i}(j)$ and storing them as $n$ sorted lists takes
  $mn\cdot T_{V}+O(mn\log m)$ time. By Corollary \ref{cor:reg_sep},
  the separation oracle for $\tilde{f}$ can be implemented in
  $n^{2}\cdot T_{V}+O(n^{2} \log(m))$ time.
\end{proof}

\section{Robust Walrasian Prices, Market Coordination and Walrasian allocations}
\label{sec:robust_walrasian_prices}

So far we focused on computing Walrasian prices. Now we turn to the other
component of Walrasian equilibrium, which is to compute the optimal allocation.
If we have only access to an
aggregate demand oracle, then computing prices is all that we can hope for,
since we have no per-buyer information (in fact, we don't even know the
number of buyers). If we have access to a value oracle, computing the optimal
allocation is possible.

To convince the reader that this is a non-trivial problem, we show that
computing an optimal allocation from Walrasian prices is at least as hard as
solving the \emph{Matroid Union Problem}. In the Matroid Union problem we are
given $m$ matroids defined over the same ground set $M_i = ([n],
\mathcal{B}_i)$ and a promise that there exist basis $B_i \in \mathcal{B}_i$
such that $[n] = \cup_i B_i$. The goal is to find those bases. Now, consider the following
mapping to the problem of computing an optimal allocation: consider $m$ agents
with valuations over a set $[n]$ of items such that $v_i = r_{M_i}$, i.e. 
the rank of matroid $M_i$ (matroid rank functions are known to be gross
substitute valuations \cite{gs_survey}). The price vector $\vc{1}$ is clearly
a vector of Walrasian prices. Finding the optimal allocation, however, involves
finding $\vc{S} = (S_1, \hdots, S_m)$ maximizing $\sum_i r_{M_i}(S_i)$.

The previous paragraph hopefully convinced the reader that finding an allocation
is not always a simple task even if we know the prices. One approach to solve
this problem is based on a modification of standard matroid union algorithms.

The second approach, which we discuss here in details, is based on convex
programming and reveals an important structural property of gross
substitute valuations that might be of independent interest. Incidentally, this
also answers an open question of Hsu et al. \cite{hsu_prices} who asked what are
the conditions for markets to be perfectly coordinated using prices. More precisely, they showed that under some genericity condition the minimal Walrasian price \textit{for a restricted class of gross substitutes} induces an overdemand at most 1 for each item. On the other hand, our argument in this section says that under the same condition \textit{almost every} Walrasian prices \textit{for any gross substitutes} have \textit{no} overdemand, i.e. the market is perfectly coordinated. This follows from the fact that the polytope of Walrasian prices have nonempty interior and that interior Walrasian price induces no overdemand (see section 6.2).

Next, we review two combinatorial lemmas that will be fundamental for the rest
of this section and the next one:

\subsection{Two combinatorial lemmas}

One of the most useful (and largely unknown) facts about gross substitutes is the
following analogue to the Unique Matching Theorem for matroids.
The version of this theorem for gross substitutes is due to Murota \cite{Murota96a,
Murota96b} and it was originally proved in the context of valuated matroids,
which are known to be equivalent to gross substitutes under a certain
transformation. We refer the reader to Lemma 10.1 in \cite{gs_survey} 
for a proof of this lemma in the language of gross substitute valuations:

\begin{lemma}[Unique Minimum Matching Lemma]\label{lemma:unique_matching}
Let $v : 2^{[n]} \rightarrow \R$ be valuation satisfying gross substitutes, $S
\subseteq [n]$, $A = \{a_1, \hdots, a_k\} \subseteq S$, $B = \{b_1, \hdots,
b_k\} \subseteq [n] \setminus S$. Consider weighted bipartite graph $G$ with left set
$A$, right set $B$ and edge weights $w_{a_i, b_j} = v(S) - v(S \cup b_j
\setminus a_j)$. If $M = \{(a_1, b_1), (a_2, b_2), \hdots, (a_k, b_k)\}$ is
the unique minimum matching in $G$, then:
$$v(S) - v(S \cup B \setminus A) = \sum_{j=1}^k v(S) - v(S \cup b_j \setminus
a_j)$$
\end{lemma}

Lastly, we state a purely combinatorial lemma that is commonly used in
conjunction with the previous lemma. We present a sketch of the proof in the
appendix and refer to \cite{gs_survey} for a complete proof.

\begin{lemma}\label{lemma:combinatorial_matching}
Let $G = (V,E,w)$ be a weighted directed graph without negative weight cycles.
Let $C$ be the cycle of
minimum number of edges among the cycles with minimum weight. Let $M := \{(u_1, v_1),
\hdots, (u_t, v_t)\}$ be a set of non-consecutive edges in this cycle, $U =
\{u_1, \hdots, u_t\}$ and $V = \{v_1, \hdots, v_t\}$. Construct a bipartite
graph $G'$ with left set $U$, right set $V$ and for each edge from $u \in U$ to $v\in
V$ in the original graph, add an edge of the same weight to $G'$. Under those
conditions, $M$ forms a unique minimum matching in $G'$. The same result holds
for a path $P$ of minimum length among all the minimum weight paths between a
pair of fixed nodes.
\end{lemma}

\subsection{Robust Walrasian Prices and Market Coordination}

Hsu et al. \cite{hsu_prices} raise the following important question: when is it
possible to find Walrasian prices that coordinate the market? A vector of
Walrasian prices $\vc{p}$ is said to coordinate the market if each agent has a
unique demanded bundle under $\vc{p}$ and those bundles clear the market. If
this happens, we say that this vector is \emph{robust}.

\begin{definition}[Robust Walrasian Prices]
A price vector $\vc{p}$ is said to be a vector of \emph{robust Walrasian prices}
for a certain market if $D(i,\vc{p}) = \{S_i\}$ and $\vc{S} = (S_1, \hdots, S_m)$
form a partition of the items. 
\end{definition}

Notice that by the Second Welfare Theorem (Lemma \ref{lemma:welfare_thm}), if
the optimal allocation is not unique, then no vector of Walrasian prices is
robust, since each vector of Walrasian prices support all the allocations. If
the optimal allocation is unique, on the other hand, then we show that a vector
of robust Walrasian prices always exists. Moreover, the set of Walrasian prices forms a
full-dimensional convex set in which all interior points are robust.

\begin{theorem}\label{thm:robust_walrasian}
If there is a unique partition $\vc{S} = (S_1, \hdots, S_m)$ maximizing $\sum_i
v_i(S_i)$, then there exist a vector $\vc{p}$ such for all $\vc{p}' \in \prod_j
[p_j - \frac{1}{2n}, p_j + \frac{1}{2n}]$ are Walrasian. In particular, the set of
Walrasian prices is a full-dimensional convex set and all price vectors in its
interior are robust Walrasian prices.
\end{theorem}

The proof involves the concept of the \emph{exchange graph}, which was first introduced
by Murota in \cite{Murota96b} and it characterizes the set of all Walrasian
prices as the dual variables of the shortest path polytope for a certain graph.
Given an optimal allocation $\vc{S} = (S_1, \hdots, S_m)$, the Second Welfare
Theorem (Lemma \ref{lemma:welfare_thm}) combined with the characterization of
gross substitute functions from Discrete Convex Analysis (Definition
\ref{def:discrete_convex}) tells us that the set of Walrasian prices can be
characterized by:
$$P = \left\{ \vc{p} \in \R^n \text{ } \left| \begin{aligned}
& v_i(S_i) \geq v(S_i \setminus j) + p_j, & \forall i \in [m],  j \in S_i \\
&v_i(S_i) \geq v(S_i \cup k) - p_k,  & \forall i \in [m], k \notin S_i \\
&v_i(S_i) \geq v(S_i \cup k \setminus j) - p_k + p_j, & \forall i \in [m], j \in
S_i, k \notin S_i
\end{aligned}\right. \right\}$$

Which is clearly a convex set defined by $O(\sum_i \abs{S_i} n) = O(n^2)$
inequalities. A nice combinatorial interpretation of this polytope is that it
corresponds to the set of potentials in a graph.

To make the construction nicer, augment the items with $m$ dummy items, one for
each buyer. The set of items becomes $[n] \cup [m]$, and the valuations are
extended to the subsets of $[n] \cup [m]$ in a way that agents completely ignore
the dummy items, i.e., for $T \subset [n] \cup [m]$, $v_i(T) = v_i(T \cap [n])$.
Also, augment $S_i$ to contain the dummy item for buyer $i$. Under this
transformation we can simplify the definition of $P$ to:

$$P = \left\{ \vc{p} \in \R^n \text{ } \left| \begin{aligned}
v_i(S_i) \geq v(S_i \cup k \setminus j) - p_k + p_j, & \forall i \in [m], j \in
S_i, k \notin S_i
\end{aligned}\right. \right\}$$
since we can represent adding and removing an item as a swap with a dummy item.
Under this transformation, construct a directed graph with one node
for each item in $[n]$. For
each $i \in [m]$, $j \in S_i$ and $k \notin S_i$, add an edge $(j,k)$ with
weight $w_{j,k} = v_i(S_i) - v_i(S_i \cup k \setminus j)$.

Since the allocation $\vc{S} = (S_1, \hdots, S_m)$ is optimal, there exists at least
one vector of Walrasian prices $\vc{p} \in P$. This guarantees that the graph
exchange graph has no negative cycles, since for any cycle $C = \{(j_1, j_2),
\hdots, (j_t, j_1)\}$, we can bound the sum of weights by 
$\sum_r w_{j_r, j_{r+1}} \geq \sum_r p_{j_r} - p_{j_{r+1}} = 0$, where the
inequality follows from the definition of $P$ and the definition of the weights.
Now we argue that the exchange graph can't contain any cycles of zero weight:

\begin{lemma} If $\vc{S}$ is the unique optimal allocation, then the exchange
graph has no cycles of zero weight.
\end{lemma}

\begin{proof}
If there were cycles of zero weight, let $C$ be the cycle of zero weight of
minimum length. Now, let $C_i = \{ (j_1, t_1), \hdots, (j_a, t_a) \}$ be the
edges $(j,t)$ in $C$ with $j \in S_i$ and (consequently) $t \notin S_i$. Now, define $S'_i =
S_i \cup \{t_1, \hdots, t_a\} \setminus \{j_1, \hdots, j_a\}$. Notice that we
performed the swaps prescribed by the cycle, so each item that moved was removed
from one set and added to another and as a result, $\vc{S'} = (S'_1, \hdots, S'_m)$ is
still a partition of the items.
Using Lemmas
\ref{lemma:combinatorial_matching} and \ref{lemma:unique_matching} we get that:
$$v_i(S'_i) - v_i(S_i) = \sum_{r=1}^a v_i(S_i \cup t_r \setminus j_r) - v_i(S_i) =
\sum_{r=1}^a w_{jt}$$
therefore:
$$\sum_i v_i(S'_i) - \sum_i v_i(S_i) = \sum_{e \in C} w_e = 0$$
and so $\vc{S'} = (S'_1, \hdots, S'_m)$ is an alternative optimal
allocation, contradicting the uniqueness of $\vc{S}$.
\end{proof}

Now we are ready to prove the Theorem \ref{thm:robust_walrasian}:

\begin{proof}[Proof of Theorem \ref{thm:robust_walrasian}]
Since we know there are no zero weight cycles and all the edge weights are
integral, all cycles have weight at least 1. Now perform the following operation: for each vertex $j \in [n]$ in the directed
graph, split it into $j_1$ and $j_2$ with an edge between $j_1$ and $j_2$ with
weight $-\frac{1}{n}$. Make all incoming edges to $i$ be incoming to $j_1$ and
all outgoing edges from $j$ to be outgoing from $j_2$. The resulting graph has
again no cycles of negative weight, since the new edges can decrease each cycle
by at most $1$.

Therefore, it is possible to find a potential in this graph. A potential of a
weighted graph with edge weights $w_{jt}$ is a function $\phi$ from the nodes to
$\R$ such that $\phi(t) \leq \phi(j) + w_{jt}$. It can be easily computed by
running a shortest path algorithm from any fixed source node and taking the
distance from source node to $j$ as $\phi(j)$. For the particular case of the graph
constructed, it is useful to take the source as one of the dummy nodes.

After computing a potential from the distance from a dummy node to each node,
define the price of $j$ as $p_j = \phi(j_2)$. By the definition of the potential
for each edge $(j,t)$ in the graph:
$$p_{t} = \phi(t_2) \leq \phi(t_1) - \frac{1}{n} \leq \phi(j_2) + w_{j,t} -
\frac{1}{n} = p_j + w_{j,t} - \frac{1}{n}.$$ This means in particular that all
inequalities that define $P$ are valid with a slack of $\frac{1}{n}$. Therefore,
changing any price by at most $\frac{1}{2n}$ still results in a valid Walrasian
equilibrium. This completes the proof of the first part of the theorem.

Since $P$ contains a cube, then it must be a full-dimensional convex body.
Finally, let's show that every price vector in the interior of $P$ is a vector
of robust Walrasian prices. By the second welfare theorem (Lemma
\ref{lemma:welfare_thm}), $S_i \in D(i,
\vc{p})$ for all $\vc{p} \in P$. Now, assume that for some point in the
interior, there is $S'_i \in D(i,\vc{p})$, $S'_i \neq S_i$. Then either: (i)
There is $j \in S'_i \setminus S_i$.  We decrease the price of $j$ by
$\epsilon$ so $S'_i$ becomes strictly better than $S_i$, i.e. $S_i \notin D(i, \vc{p}
-\epsilon \vc{1}_j)$, which contradicts the second welfare theorem since $\vc{p}
- \epsilon \vc{1}_j \in P$. (ii) There is $j \in S_i \setminus S'_i$. 
We increase the price of $j$ by $\epsilon$ so $S'_i$ becomes strictly better than
$S_i$, i.e. $S_i \notin D_i(i, \vc{p} + \epsilon \vc{1}_j)$, which again
contradicts the second welfare theorem since $\vc{p} + \epsilon \vc{1}_j \in P$.
\end{proof}

\subsection{Computing Optimal Allocation}

Theorem \ref{thm:robust_walrasian} guarantees that if the optimal allocation is
unique, then the set of Walrasian prices has large volume. Since the set of
Walrasian prices corresponds to the set of minimizers of the market potential
$f(\vc{p})$, then there is a large region where zero is the unique subgradient
of $f$. In such situations, convex optimization algorithms are guaranteed to
eventually query a point that has zero subgradient. The point $\vc{p}$ queried
corresponds to a set of Walrasian prices and the optimal allocation can be
inferred from the subgradient (recall Theorem \ref{thm:general_walrasian_non_empty}).

Our strategy is to perturb the valuation functions in such a way that the
optimal solution is unique and that it is still a solution of the original
problem. It is important for the reader to notice that this is a different
type of perturbation than the one used in previous sections. While previously we
perturbed the objective of the dual program, here we are effectively perturbing
the objective of the primal program. One major difference is that if we perturb
the objective of the dual, we can still compute the subgradient of $\hat{f}$ in
\ref{perturbed_convex_f} using the aggregate demand oracle. If we perturb the
objective of the primal, we no longer can compute subgradients using the
aggregate demand oracle. With value oracles, however, this is still possible to
be done.

To perturb the primal objective function, we use the isolation lemma, a standard
technique to guarantee a unique optimum for combinatorial problems.

\begin{lemma}[Isolation Lemma \cite{IsolationLemma}]
Let $\vc{w} \in [N]^n$ be a random vector where each coordinate $w_i$ is choosen
independently and uniformly over $[N]$. Then for any arbitrary family
$\mathcal{F}\subseteq 2^{[n]}$, the problem $\max_{S \in \mathcal{F}} \sum_{i
\in S} w_i$ has a unique optimum with probability at least $1-n/N$.
\end{lemma}

For our application, we would like a unique optimum to the problem of maximizing
$\sum_{i=1}^m v_i(S_i)$ over the partition $(S_1,\hdots,S_m)$. To achieve
this, we first replace $v_i$ by $\tilde{v_i}(S)=B
v_i(S) + w_i(S)$ where $B$ is some big number to be determined and
$w_i(j)$ is set as in the isolation lemma (with $N$ to be determined as well).

\begin{lemma}\label{lem:unique}
By setting $B = 2nN$, $N=mn^{O(1)}$ and sampling $w_i(j)$ uniformly from $[N]$
for each $i \in [m]$ and $j \in [n]$, then with probability $1-1/n^{O(1)}$,
there is an unique partition $(S_1, \hdots, S_m)$ maximizing $\sum_{i=1}^m
\tilde{v_i}(S_i)$, for $\tilde{v_i}(S) = B \cdot v_i(S) + w_i(S)$.
\end{lemma}
\begin{proof}
Let $(S_1',\hdots,S_m')$ be an optimal solution w.r.t the original problem. We
first show that any suboptimal partition $(S_1,\hdots,S_m)$ cannot be optimal
for the perturbed problem. Since $v_i$ assume integer values, we have
$\sum_{i=1}^m v_i(S_i')\geq  \sum_{i=1}^m v_i(S_i) +1$. Now:
$$\sum_{i=1}^{m}\tilde{v_{i}}(S_{i}')	\geq
B+\sum_{i=1}^{m}\tilde{v_{i}}(S_{i})+\sum_{i=1}^{m}w_{i}(S_{i})-\sum_{i=1}^{m}w_{i}(S_{i}')
\geq
B+\sum_{i=1}^{m}\tilde{v_{i}}(S_{i})-nN>\sum_{i=1}^{m}\tilde{v_{i}}(S_{i}).$$
which shows that a suboptimal solution to the original problem cannot be optimal
for the perturbed one.

Now, consider all partitions  $(S_1',\hdots,S_m')$ that are optimal for the
original problem and identify each optimal partition with a
subsets of $[mn] = \{(i,j); i \in [m], j \in [n]\}$ in the natural way: add
$(i,j)$ to the subset if $j \in S'_i$. This family of subsets corresponds to
$\mathcal{F} \subseteq 2^{[mn]}$ in the statement of the Isolation Lemma and
$w_i(j)$ corresponds to $\vc{w}$. The result then follows from applying that
lemma.
\end{proof}

The strategy to find an optimal allocation is to perturb the valuation
functions, then search for an interior point in the set of minimizers of the
market potential function $f$. When we find such a point we can obtain a vector
of Walrasian prices for the original market by rounding and the optimal
allocation by inspecting the subgradient. To get the desired running time, we
need to apply those ideas to the regularized potential function $\tilde{f}$
(see \ref{reg_convex_f} in Section \ref{sec:gs_value_oracle}) instead of the
original one. To apply this to the regularized potential we need an extra lemma:

\begin{lemma}\label{lem:perturbgivesol}
Let $\vc{p}$  be an interior point of the set of minimizers of the regularized
potential function $\tilde{f}$, then the allocation $(S_1^*,\hdots,S_m^*)$ 
produced by the \textsc{AllGreedy} algorithm
in Lemma \ref{lem:regularize_lemma} is an equilibrium allocation.
\end{lemma}

\begin{proof} If $(S_1^*,\hdots,S_m^*)$ is a partition over the items, then it
is an optimal allocation by part 2 of Lemma \ref{lem:regularize_lemma}. To show
that it is a partition, observe that since
$(S_1^*, \hdots, S_m^*)$ is a maximizer of $\sum_i v_i(S_i) - \vc{p}(S_i)$
subject to $\sum_i \abs{S_i} = n$, then we can use it to build a subgradient
$\vc{g}$ of $\tilde{f}$ such that $g_j = -1 + \abs{\{i; j \in S_i^*\}}$. Since
$\vc{p}$ is an interior point of the set of minimizers, the subgradient must be
zero and therefore $\abs{\{i; j \in S_i^*\}} = 1$ for all $j$.
\end{proof}

\begin{theorem} For gross substitute valuation, we can find a Walrasian
equilibrium, i.e. allocation and prices, in time $O((mn + n^3) T_V  \log(nmM) +
n^3 \log^{O(1)}(mnM) )$.
\end{theorem}

\begin{proof}
Use Lemma \ref{lem:unique} to perturb the valuation functions and obtain
$\tilde{v}_i$ which are still gross substitutes (since they are the sum of a gross
substitute valuation and an additive valuation) and there is an unique optimal
allocation. By Lemma \ref{thm:robust_walrasian}, the set of minimizer of the
market potential function $f$ contains a box of width $1/n$. Since the set of
minimizers of the market potential $f$ is contained in the set of minimizers of
the regularized potential $\tilde{f}$, then its set of minimizers also
contains a box of width $1/n$. We also know that
it is contained in the box $[-Mmn^{O(1)}, Mmn^{O(1)}]^n$ since
$\abs{\tilde{v}_i(S)} \leq O(Mmn^{O(1)})$.

If we apply the algorithm in Theorem \ref{thm:lsw} with
$\delta = O(1/n^{O(1)})$ to optimize the regularized potential $\tilde{f}$
then we are guaranteed to query a point in the
interior of minimizers as otherwise the algorithm would certify that there is
$\vc{a}$ with $\norm{\vc{a}}_2 = 1$ such that $\max_{\vc{p} \in P} \vc{a} \cdot
\vc{p} - \min_{\vc{p} \in P } \vc{a} \cdot \vc{p} \leq 1/n^{O(1)}$, where $P$ is
the set of minimizers.

Finally, we can obtain the optimal allocation for the perturbed using Lemma
\ref{lem:regularize_lemma}, which is an optimal allocation to the original
market according to Lemma \ref{lem:unique}.
\end{proof}

\section{Combinatorial approach to Walrasian Equilibrium for Gross Substitutes}
\label{sec:murota}

In a sequence of two foundational papers \cite{Murota96a, Murota96b}, Murota
shows that the assigment problem for \emph{valuated matroids}, a class of
functions introduced by Dress and Wenzel \cite{DressWenzel} can be solved in
strongly polynomial time. We show how this algorithm can be used to obtain an
$\tilde{O}(nm + n^3)$ strongly polynomial time algorithm for problem of
computing a Walrasian equilibrium for gross substitute valuations. Our
contribution is two-fold: first we map the Walrasian equilibrium problem on gross
substitute valuations to the assignment problem on valuated matroids and analyze
its running time. The straightforward mapping allows us to obtain a strongly
polynomial time algorithm with running time $O(mn^3 \log(m+n))$. We note that a
different way to reduce the Walrasian equilibrium problem to a standard problem
in discrete convex analysis is to map it to the $M$-convex submodular flow
problem as done in Murota and Tamura \cite{murota_tamura}. We choose to
reduce it to the assignment problem on valuated matroids since its running time
is simpler to analyze.

Inspired by our
$\widetilde{O}(mn + n^3)$ algorithm, we revisit Murota's algorithm and propose
two optimizations that bring the running time down to $O((mn + n^3) \log(m+n))$.
Murota's algorithm works by computing augmenting paths in a structure known as
the \emph{exchange graph}. First we show that for the Walrasian equilibrium
problem, this graph admits a more succint representation. The we propose a data
structure to amortize the cost of some operations across all iterations.

In section \ref{subsec:valuated_matroids} we define valuated matroids and the
assignment problem for valuated matroids. Then we describe Murota's algorithm
for this problem. We also discuss the relation between the assignment problem
for valuated matroids and the welfare problem for gross substitutes. The goal of
subsection \ref{subsec:valuated_matroids} is to provide the reader with the
historical context for this result. 

The reader insterested solely in the welfare problem is welcome to skip to
Section \ref{subsec:algorithm} which can be read independently, without any
mention to valuated matroids or the assignment problem. A complete proof is
given in that section.

\subsection{The assignment problem for valuated
matroids}\label{subsec:valuated_matroids}

A \emph{valuated matroid} is an assignment of weights to basis of a matroid
respecting some valuated analogue of the exchange property.

\begin{definition}[Valuated matroid]
Let $\B$ be the set of basis of a matroid $\M = (V, \B)$. A valuated
matroid is a function $\omega: \B \rightarrow \R \cup \{\pm \infty\}$ such that for
all $B, B' \in \B$ and $u \in B \setminus B'$ there exists $v \in B' \setminus
B$ such that:
$$\omega(B) + \omega(B') \leq \omega(B \cup v \setminus u) + \omega(B' \cup u
\setminus v)$$
\end{definition}

Valuated matroids are related to gross substitutes by the following one-to-one
correspondence. We refer the reader to Lemma 7.4 in \cite{gs_survey} for a proof.

\begin{proposition}\label{prop:equiv_valuated_matroids}
A valuation function $v:2^{[n]} \rightarrow \R$ is a gross substitutes valuation
function iff $\omega: {[2n] \choose n} \rightarrow \R$ defined by $\omega(S) =
v(S \cap [n])$ is a valuated matroid defined over the basis of the $n$-uniform
matroid on $2n$ elements.
\end{proposition}

Now, we are ready to define the assignment problem for valuated matroids:

\begin{definition}[Valuated matroid assignment problem]
Given two sets $V_1, V_2$, matroids $\M_1 = (V_1, \B_1)$ and $\M_2 = (V_2,
\B_2)$ of the same rank, valuated matroids 
$\omega_1: \B_1 \rightarrow \R$ and $\omega_2: \B_2 \rightarrow \R$ 
and a weighted bipartite graph $G = (V_1 \cup V_2, E, w)$, find a
matching $M$ from $V_1$ to $V_2$ maximizing:
$$w(M) + \omega_1(M_1) + \omega_2(M_2)$$
where $M$ is a subset of edges of $E$ forming a matching and $M_1$ and $M_2$ are
the sets of endpoints of $M$ in $V_1$ and $V_2$ respectively.
\end{definition}

Murota gives two strongly-polynomial time algorithms based on network flows for the
problem above in \cite{Murota96b} -- the first based on cycle-cancellations and
the second based on flow-augmentations. Although the running time is not
formally analyzed in his paper, it is possible to see that his algorithm (more
specifically the algorithm \emph{Augmenting algorithm with potentials} in
Section 3 of \cite{Murota96b}) has running time $\tilde{O}(\rank(\M_i)\cdot
(\abs{E} + \rank(\M_i)\cdot( \abs{V_1} + \abs{V_2})))$.

First, we show a simple reduction from the welfare problem for gross substitutes
to this problem. Given $m$ gross substitute valuation functions $v_i:2^{[n]}
\rightarrow \R$, define the following instance of the valuated matroid
assignment problem: define the first matroid as ${[mn] \choose n}$ i.e. the
$n$-uniform matroid on $mn$ elements. Interpret the elements of $[nm]$ as ``the
allocation of item $j$ to agent $i$" for each pair $(i,j)$. Following this
interpretation, each $S \subseteq [nm]$ can be seen as $S = \cup_{i=1}^m S_i$
where $S_i$ are the elements assigned to agent $i$. This allows us to define for
each $S \in {[mn] \choose n}$, $\omega_1(S) = \sum_i v_i(S_i)$. For the second
matroid, let $\M_2 = ([n], 2^{[n]})$ and $\omega_2(S) = 0$ for all $S$. Finally,
define the edges of $E$ such that for each $j \in [n]$, the $j$-th element of
$[n]$ are connected to the element $(i,j)$ in $[mn]$ for each $i$.

One needs to prove that $\omega_1$ satisfies the properties defining a
valuated matroid, but this can be done using the transformations
described in \cite{gs_survey}. We omit this proof since we are giving a
self-contained description of the algorithm in the next section. 

In the construction shown below, $\abs{V_1} = \abs{E} = nm$, $\abs{V_2} = n$ and
$\rank(\M_1) = \rank(\M_2) = n$. This leads to an $\tilde{O}(m \cdot n^3)$
strongly polynomial time algorithm for the welfare problem.

\subsection{Gross substitutes welfare problem in $\tilde{O}(nm + n^3)$
time}\label{subsec:algorithm}

In this section we give a self-contained description of a specialized
version of Murota's algorithm for the gross substitutes welfare problem and
show that the running time of $\tilde{O}(m \cdot n^3)$ can be improved to
$\tilde{O}(nm + n^3)$. Murota's algorithm for the case of generic
valuated matroids can be quite complicated. Since the underlying matroids are
simple (uniform matroids) and the functions being optimized have
additional structure, it is possible to come up with a simpler algorithm. Our
presentation also makes the algorithms accessible to the reader not familiar
with discrete convex analysis and the theory of valuated matroids.

We consider the setting in which a set $[n]$ of items needs to be allocated to
$m$ agents with monotone valuation functions $v_i : 2^{[n]}\rightarrow \R$ 
satisfying the gross substitutes condition. Consider the intermediary
problem of computing the optimal allocation for the first $k$ items (in some
arbitrary order):
\begin{equation}\tag{$I_k$}\label{eq:intermediary_problem}
\max_i v_i(S_i) \quad \text{ s.t. } \quad \cup_i S_i \subseteq [k] := \{1,2, \hdots,
k\} \text{ and } S_i \cap S_j = \emptyset \text{ for } i \neq j\end{equation}
The central idea of the algorithm is to successively solve $I_1, I_2, \hdots,
I_n$ using the solution of $I_k$ to compute $I_{k+1}$. We will show how a
solution to $(I_{k+1})$ can be computed from a solution of $(I_k)$ via a
shortest path computation in a graph with $\tilde{O}(m + n^2)$ edges.

A solution for problem $(I_k)$ consists on an allocation $S = (S_1, \hdots, S_m)$ of
items in $[k]$ to agents $1, \hdots, m$ and a price vector $p_1, \hdots, p_k$
that certifies that the allocation is optimal. Since optimality for gross
substitutes can be certified by checking that no agent wants to add an item,
remove an item or swap an item (Defintion \ref{def:discrete_convex})
then $S, p$ need to satisfy the
following conditions for every agent $i$ and every $j \notin S_i$
and $j' \in S_i$:
\begin{equation}\tag{\textsc{Add}}\label{eq:add}
v_i(S_i \cup j) - p_j \leq v_i(S_i)
\end{equation}
\begin{equation}\tag{\textsc{Remove}}\label{eq:remove}
v_i(S_i \setminus j') + p_{j'} \leq v_i(S_i)
\end{equation}
\begin{equation}\tag{\textsc{Swap}}\label{eq:swap}
v_i(S_i \cup j \setminus j') - p_j + p_{j'} \leq v_i(S_i)
\end{equation}

\subsubsection{Exchange graph}

The first step to solve $(I_{k})$ is to build a combinatorial
object called the \emph{exchange graph} using the solution of $(I_{k-1})$,
expressed as a pair $S, p$. We define a weighted directed graph on $V = [k] \cup
\{\phi_1, \hdots, \phi_m\}$. Intuitively, we can think of $\phi_i$ as an ``empty
spot" in the allocation of agent $i$. We add edges as follows:
\begin{itemize}
\item $(t,j)$ for all items $t$ and $j$ not acquired by the same agent under
$S$. If $i$ is the agent holding item $t$, then the edge represents the
change in utility (under price $p$) for agent $i$ to swap his item $t$
by $j$: $$w_{t,j} = v_i(S_i) - v_i(S_i \cup j \setminus t) + p_j - p_t$$
\item $(\phi_i,j)$ for all items $j \notin S_i$. It represents the change in
utility for $i$ to add item $j$: $$w_{\phi_i,j} = v_i(S_i) - v_i(S_i \cup j) +
p_j$$
\end{itemize}
Notice that problem $(I_{k-1})$ only defines prices for $1, \hdots, k-1$. So for
the construction above to be well defined we need to define $p_k$. We will set
$p_k$ in a moment, but before that, note that for all the edges not involving
$k$, $w \geq 0$ which follow by the fact that $p$ is a certificate of optimality
for $S$ and therefore conditions \ref{eq:add} and \ref{eq:swap} hold. Finally,
notice that there are only directed edges from $k$ to other nodes, so $p_k$
always appear with positive sign in $w$. So, we can set $p_k$ large enough such
that all egdes have non-negative weights, in particular, set:
$$p_k = \max \left\{ \max_{i \in [m]; t \in S_i} [ v_i(S_i \cup k \setminus t) 
- v_i(S_i) - p_t], \text{ } \max_{i \in [m]} [v_i(S_i \cup k) - v_i(S_i)] \right\} $$

\subsubsection{Updating prices and allocations via shortest
path}\label{subsec:update}

After the exchange graph is built, the algorithm is trivial: compute the
minimal-length shortest path from some $\phi_i$  (i.e. among all paths of
minimum weight between $\phi_i$ and $k$, pick the one with minimum length).
Since the edges are non-negative, the shortest path can be computed in the order
of the number of edges using Dijkstra's algorithm in $O(\abs{E} \cdot \log
\abs{E})$ where $\abs{E}$ is the number of edges in the graph. Dijkstra's
algorithm can be easily modified to compute the minimum weight path with
shortest length using the following idea: if weights are integers, then
substitute weights $w_{ij}$ by $w_{ij} - \epsilon$. In the end, round the
solution up. Or more formally, run Dijkstra in the ordered ring $(\Z^2, +, <)$
with weights $(w_{ij},1)$ where $+$ is the componentwise sum and $<$ is the
lexicographic order.

Let $P$ be the path output by Dijkstra. Update the allocation by performing the
swaps prescribed by $P$. In other words, if edge $(t,j) \in P$ and $t \in S_i$,
then we swap $t$ by $j$ in $S_i$. Also, if edge $(\phi_i,j) \in P$ we add $j$ in
$S_i$. Formally, let $(t_r,j_r)_{r=1..a}$ be all the edges in $P$ with $t_r \in S_i$ or
$t_r = \phi_i$. Then we update $S_i$ to $S_i \cup \{j_1,\hdots,j_a\} -
\{t_1,\hdots, t_a\}$.

The execution of Dijkstra also produces a certificate of optimality of the
shortest path in the form of the minimum distance from some $\phi_i$ to any
given node. So, there is a distance function $d$ such that
$$d(\phi_i) = 0, \qquad d(j) \leq d(\phi_i) + w_{\phi_i,j}, \qquad d(j) \leq
w_{t,j} + d(t) $$
Moreover, for all edges $(t,j)$ and $(\phi_i, j)$ in the shortest path $P$, this
holds with equality, i.e.: $d(j) = d(\phi_i) + w_{\phi_i, j}$ and
$d(j) = w_{jt} + d(t)$. Update the price of each item $j$ from $p_j$ to $p_j -
d(j)$.

\subsubsection{Running time analysis}

Before we show that each iteration produces an optimal pair of allocation $S$
and prices $p$ for problem $(I_k)$ we analyze the running time.

The exchange graph for problem $(I_k)$ as previously described has $O(mk + k^2)$
edges. Running Dijkstra's algorithm on this graph has running time $\tilde{O}(mk + k^2)$
for $(I_k)$, which corresponds to $\tilde{O}(mn^2 + n^3)$ time overall.

In order to get the overall running time down to $\tilde{O}(mn + n^3)$ we need
one extra observation. Since we want to compute the shortest path from any of
the $\phi_i$ nodes to $k$, we can collapse all $\phi_i$ nodes to a single node
$\phi$. Now, for any given node $j$:
$$w_{\phi,j} = \min_i w_{\phi_i, j} = p_j + \min_i [ v_i(S_i) - v_i(S_i \cup
j)]$$
Now, the graph is reduced to $O(k^2)$ edges for problem $(I_k)$. So, Dijkstra
can be computed in $\tilde{O}(k^2)$. We are only left with the task to compute
$w_{\phi,j}$. Our task is to compute $\min_i [ v_i(S_i) - v_i(S_i \cup
j)]$. This can be divided in two parts:
\begin{enumerate}  
\item \emph{active agents}: the minimum among the agents $i$ for
  which $S_i \neq \emptyset$. There are at most $k$ of those, so we can iterate
  over all of them and compute the minimum explicitly;
\item \emph{inactive agents}: the minimum over all
  agents with $S_i = \emptyset$. In order to do so we maintain the following
  data structure: in the first iteration, i.e. in $(I_1)$, we compute
  $v_i(\{j\})$ for every $i,j$ (which takes $O(mn)$ time) and keep for each item
  $j$ a sorted list $\texttt{L}_j$ in decreasing order of $v_i(j)$ for all $i$.

  In the end of each iteration, whenever an innactive agent $i$ becomes active
  (i.e. we allocate him an item), we remove them from $\texttt{L}_j$ for all
  $j$. This operation takes $O(n)$ time to go over all lists.

  Now, once we have this structure, we can simply compute the minimum among the
  innactive buyers $\min_i [ v_i(S_i) - v_i(S_i \cup j)] = - \max_i v_i(\{j\})$
  we simply look the minimum element of the list $\texttt{L}_j$. Therefore, even
  though we need to pay $O(mn)$ time in $(I_1)$. In each subsequent iteration we
  pay only $O(n)$ to update lists $\texttt{L}_j$ and then we can make query the
  value of $w_{\phi,j}$ in constant time.
\end{enumerate}
This leads to a running time of $O(mn)$ in $(I_1)$ and $\tilde{O}(n + k^2)$ in
each subsequent iteration, leading to an overall running time of $\tilde{O}(mn +
n^3)$. We also note that for each edge we build of the graph, we query the
value oracle once. So the oracle complexity is $O(nm + n^3)$ value oracle calls.

\subsubsection{Correctness}

We are left to argue that the solution $(S,p)$ produced by the algorithm is indeed
a valid solution for problem $(I_k)$. This can be done by checking that the
price vector $p$ obtained certifies the optimality of $S$. The main ingredients
for the proof are Lemmas \ref{lemma:unique_matching} and
\ref{lemma:combinatorial_matching}. We encourage the reader the revisit the
statement of those lemmas before reading the proof of the following theorem.

\begin{theorem}
Let $S^k_i, p^k_j$ be the solution of problem $(I_k)$, then for all agent $i$
and all $S' \subseteq [k]$, $$v_i(S^k_i) - p^k(S^k_i) \geq v_i(S') - p^k(S').$$
\end{theorem}

\begin{proof}
Let $S^{k-1}_i, p^{k-1}_j$ be the solution of problem $(I_{k-1})$ and let
$S^k_i, p^k_j$ be the solution of problem $(I_k)$. If $d(\cdot)$ is the distance
function returned by Dijkstra in $(I_k)$, then $p^k_j = p^{k-1}_j - d(j)$ and we
also know that $\tilde{w}_{jt} = w_{jt} + d(t) - d(j) \geq 0$ and
$\tilde{w}_{\phi_i,j} =w_{\phi_i,j} - d(j) \geq 0$ by the observation in the end
of the Section \ref{subsec:update}.  This implies that for all $i$ and $j \notin
S^{k-1}_i$, it holds that:
$$\tilde{w}_{t,j} = v_i(S_i^{k-1}) - v_i(S_i^{k-1} \cup j \setminus t) 
+ p^k_j - p^k_t \geq 0 \quad \text{ and } \quad \tilde{w}_{\phi_i,j} = 
v_i(S^{k-1}_i) - v_i(S^{k-1}_i \cup j) + p^k_j \geq 0$$
This means that properties \ref{eq:add} and \ref{eq:swap} are satisfied. To see
that \ref{eq:remove} is also satisfied for $j < k$ since $v_i(S^{k-1}_i) \geq
v_i(S^{k-1}_i \setminus j) + p^{k-1}_j$ for all $j \in S^{k-1}_i$. Since $p^k_j
\leq p^{k-1}_j$, this condition must continue to hold.

This means that under the price vector $p^k$ the bundles $S_i^{k-1}$ is still
the demanded bundle by agent $i$ (by Definition \ref{def:discrete_convex}).
This means in particular that for all $S' \subseteq [k]$:
$$v_i(S_i^{k-1}) - p^k(S_i^{k-1}) \geq v_i(S') - p^k(S')$$
Now, let $(t_r, j_r)_{r=1..a}$ be the set of edges in the path $P$ outputted by
Dijkstra such $t_r \in S_i^{k-1}$. Then $S^{k}_i = S^{k-1}_t \cup \{j_1, \hdots,
j_a\} \setminus \{ t_1, \hdots, t_a\}$.

Using Lemma \ref{lemma:combinatorial_matching}, we note that $(t_1,r_1), \hdots,
(t_a, r_a)$ is an unique minimum matching in the sense of Lemma
\ref{lemma:unique_matching}. Therefore:
$$v_i(S^k_i) - v_i(S^{k+1}_i) = \sum_{r=1}^a v_i(S^k_i) - v_i(S^k_i \cup j_r
\setminus t_r)$$
Summing $-p^k(S^k_i) + p^k(S^{k+1}_i)$ on both sides, we get:
$$[ v_i(S^k_i) - p^k(S^k_i) ] - [ v_i(S^{k+1}_i) - p^k(S^k_i) ] = \sum_{r=1}^a
\tilde{w}_{t_r,j_r} = 0$$
Therefore:
$$ v_i(S^k_i) - p^k(S^k_i) = v_i(S^{k+1}_i) - p^k(S^k_i) \geq  v_i(S') -
p^k(S')$$ 
as desired.
\end{proof}

\subsubsection{Descending Auction View}

One can reinterpret the procedure above as a descending auction. Initially
all items very large price (say like the price set for $p_k$ in the beginning of
phase $k$). Each shortest path computation produces a distance function $d$ that
dictates how each price should decrease. Indeed, they monotonically decrease
until we reach a Walrasian equilibrium.

We note that it is important in this algorithm that we compute in each step both
a primal and a dual solution. Without the dual solution (the price vector), it
is still possible to carry out the shortest path computation, but since the
edges in the path can have mixed signs, Dijkstra's algorithm is no longer
available and one needs to pay an extra factor of $n$ to run Bellman-Ford's
algorithm.

\bibliographystyle{alpha}
\footnotesize
\bibliography{sigproc}
\normalsize
\appendix

\section{Missing Proofs}\label{appendix:missing_proofs}

\begin{proofof}{Lemma \ref{lemma:welfare_thm}}
Let $\vc{y}=(y^{(1)},y^{(2)},\ldots,y^{(m)})$ be a valid allocation that
achieves the optimal social welfare. Then since $x^{(i)}\in D(i,\vc{p})$,
we have
$$v_{i}(x^{(i)})-\vc{p} \cdot x^{(i)}\geq v_{i}(y^{(i)})-\vc{p}\cdot y^{(i)}.$$
Summing up, we get
$$\sum_{i}\left(v_{i}(x^{(i)})-\vc{p}\cdot x^{(i)}\right)
\geq
\sum_{i}\left(v_{i}(y^{(i)})-\vc{p}\cdot y^{(i)}\right).$$
the crucial observation is that $\sum_{i}\vc{p}\cdot x^{(i)}=\sum_{i}\vc{p}\cdot y^{(i)}=\sum_{j}p_{j}s_{j}$.
herefore the inequality above simplifies
to
$$
\sum_{i}v_{i}(x^{(i)})\geq\sum_{i}v_{i}(y^{(i)}),
$$
i.e. the social welfare of $\vc{x}$ is at least that of $\vc{y}$. But since
$\vc{y}$ gives the optimal social welfare, we must then have equality
and $\vc{x}$ also achieves the optimum.

For the second part, notice that since the last equation holds with equality,
then the previous equations should also hold, therefore: 
$ \sum_{i}v_{i}(x^{(i)})-\vc{p}\cdot x^{(i)} =
\sum_{i}v_{i}(y^{(i)})-\vc{p}\cdot y^{(i)}$, which says that $x_i$
is also a favorite bundle of $i$ under price vector $\vc{p}$. Therefore,
$(\vc{x}, \vc{p})$ form a Walrasian equilibrium.
\end{proofof}

\begin{proofof}{Lemma \ref{lem:eqismin}}
  If $(\vc{p}^{\eql},\vc{x})$ is a Walrasian equilibrium it is straightforward
  to check that setting $\vc{p} = \vc{p}^{\eql}$, $u_i = \max_{x_i \in
  \dbracket{s}}   v_i(x) - \vc{p}^{\eql}\cdot x$ and $z_{i,x} = 1$ when $x =
  x^{(i)}$ and zero otherwise, we have a primal dual pair of feasible solutions
  with the same value.
  Conversely, if the primal program has an integral solution, the definition of
  Walrasian equilibrium can be obtained from the complementarity conditions.

  If the primal program has an optimal integral solution $\vc{x}$, then
  for every solution $(\vc{p}, \vc{u})$ to the dual program:
  $\sum_i v_i(x^{(i)}) = \sum_i u_i + \vc{p} \cdot \vc{s} \geq
  \sum_i v_i(x^{(i)}) + \vc{p} \cdot x_i \geq \sum_i v_i(x^{(i)})$ and
  therefore all inequalities must hold with equality, so in particular $x^{i}
  \in D(v_i, \vc{p})$, so $\vc{p}$ is a vector of Walrasian prices. Conversely,
  if $\vc{p}$ is a vector of Walrasian prices then $(\vc{x}, \vc{p})$ is
  Walrasian equilibrium by the Second Welfare Theorem (Lemma
  \ref{lemma:welfare_thm}). Therefore
  by setting $u_i = v_i(x^{(i)}) - \vc{p} \cdot x^{(i)}$ we obtain a dual feasible
  solution such that $\sum_i u_i + \vc{p} \cdot \vc{s} = \sum_i v_i(x^{(i)})$
  and therefore $(\vc{p}, \vc{u})$ is an optimal dual solution.
\end{proofof}

\begin{proofof}{Corollary \ref{cor:reg_minimizer}}
Let $\vc{p}^{\eql}$ and $\vc{S}=(S_1,S_2,\ldots,S_m)$ be an equilibrium
price and allocation.  Consider the following chain of inequalities:
$$\sum_i v_i(S_i) \leq \tilde{f}(\vc{p}^*) \leq
  \tilde{f}(\vc{p}^{\eql}) \leq f(\vc{p}^{\eql}) = \sum_i v_i(S_i)$$
Where the first inequality follows from the definition of $\tilde{f}$, the
second from the fact that $\vc{p}^*$ is a minimizer of $\tilde{f}$, the third
follows from the fact that $\tilde{f} \leq f$ for all prices $\vc{p}$, since
$f$ is a maximization over all $S_i \subseteq [n]$ and $\tilde{f}$ is a
maximization over all subsets whose cardinality is exactly $[n]$. The last
inequality follows from the fact that $\vc{p}^\eql$ is an equilibrium. This
implies that all inequalities should hold with equality, in particular, since
$\tilde{f}(\vc{p}^*) = \sum_i v_i(S_i)$, then it must be that:
$$\max_{S \subseteq [n]} v_i(S) - \vc{p}^*(S) - \gamma \cdot \abs{S} = v_i(S_i)
- \vc{p}^*(S_i) - \gamma \cdot \abs{S_i}$$
In particular, $S_i \in D(i, \vc{p}^* + \gamma \cdot \vc{1}_{[n]})$.

The other direction is similar. We have $\tilde{f}(\vc{p}^{\eql})=\sum_i v_i(S_i)= f(\vc{p}^{\eql})$ and for any price $p$,
$$\tilde{f}(\vc{p})=f(\vc{p}+\gamma\cdot\vc{1}_{[n]})\geq f(\vc{p}^{\eql})=\tilde{f}(\vc{p}^{\eql})$$
which shows that any Walrasian price $\vc{p}^{\eql}$ minimizes $\tilde{f}$.
\end{proofof}

\begin{proofof}{Lemma \ref{lemma:combinatorial_matching} (sketch)}
Assume that the bipartite graph has a different matching with
total weight not larger than the one presented. Then it is possible to
construct either a cycle of weight less than $C$ or a cycle of the same weight
with smaller number of edges.

Let $M'$ be an alternative matching between $U$ and $V$ of weight at most
the weight of $M$. If $M'$ has smaller
weight, replace $M$ by $M'$ and $C \cup M' \setminus M$ is a collection of
cycles with total weight smaller than the weight of $C$. Since all cycles have
non-negative weight, one of the cycles must have weight less than $C$,
contradicting the fact that $C$ is a minimum weight cycle.

Now, if $M$ and $M'$ have the same weight, consider the following family of
cycles: for each edge $e = (u', v') \in M'$, construct a cycle $C_e$ composed
of edge $e$ and the path from $v'$ to $u'$ in $C$ (in other words, we use $e$ to
shortcut $C$). There is an integer $k \leq t$ such that the collection
of cycles $C_e$ uses in total: one of each edge from $M'$, $k-1$ of each
edge from $M$ and $k$ of each edge from $C \setminus M$. So the average weight
is at most the weight of $C$. Since the $C_e$ cycles have strictly less edges than
$C$, there should be a cycle with fewer edges than $C$ and weight at most $C$,
which again contradicts the choice of $C$.

The argument for paths is analogous.
\end{proofof}

\end{document}